\documentclass[11pt]{article}
\usepackage[a4paper, margin=1in]{geometry}

\usepackage{algorithm}
\usepackage{algpseudocode}
\usepackage{hyperref}
\usepackage{booktabs}
\usepackage{multirow}
\usepackage{enumitem}

\usepackage{tabularx}

\usepackage{pgfplots}
\pgfplotsset{compat=1.18}

\definecolor{acmDarkBlue}{RGB}{0,114,178}
\definecolor{acmGreen}{RGB}{0,158,115}
\definecolor{acmPink}{RGB}{204,121,167}
\definecolor{acmOrange}{RGB}{213,94,0}
\definecolor{acmYellow}{RGB}{240,228,66}
\definecolor{acmLightBlue}{RGB}{86,180,233}

\pgfplotsset{
    cycle list={
        {acmDarkBlue, mark=*},
        {acmGreen, mark=square*},
        {acmPink, mark=triangle*},
        {acmOrange, mark=diamond*},
        {acmYellow, mark=o},
        {acmLightBlue, mark=star}
    }
}

\newcommand*\circled[1]{\tikz[baseline=(char.base)]{
            \node[shape=circle,draw,inner sep=1.5pt] (char) {#1};}}
\newcommand*\dottedcircled[1]{\tikz[baseline=(char.base)]{
            \node[shape=circle,draw,dotted,inner sep=1.5pt] (char) {#1};}}

\newcommand\doubleplus{+\kern-2ex+\kern0.8ex}

\usetikzlibrary{arrows,fit,positioning,shapes,chains,fit,shapes,calc,backgrounds}

\usepackage{graphicx}
\graphicspath{ {./images/} }

\usepackage{tcolorbox}

\usepackage{amssymb,amsmath,amsthm}

\usepackage{subcaption}

\usepackage{setspace}

\usepackage{booktabs}
\usepackage{arydshln} 
\usepackage{makecell}

\usepackage{thmtools,thm-restate}
\usepackage{url}

\usepackage[numbers]{natbib}

\newtcolorbox{problemBox}{
  colback=white,
  colframe=black,
  boxrule=0.5pt,
  arc=2pt,
  left=4pt,
  right=4pt,
  top=4pt,
  bottom=4pt
}

\usepackage{xcolor, soul}

\usepackage{caption}

\newtheorem{definition}{Definition}[section]
\newtheorem{theorem}[definition]{Theorem}
\newtheorem{corollary}[definition]{Corollary}

\newtheorem{example}[definition]{Example}

\newtheorem{prop}[definition]{Proposition}

\usepackage{authblk}
\usepackage{orcidlink}



\begin{document}

\title{Near-Feasible Stable Matchings: Incentives and Optimality\thanks{An extended abstract of this paper appeared in the Proceedings of AAMAS 2026 \cite{aamasfixtures}.}}
\author[ ]{Frederik Glitzner \orcidlink{0009-0002-2815-6368}}
\affil[ ]{School of Computing Science, University of Glasgow, Glasgow G12 8QQ, UK}
\affil[ ]{\normalfont \texttt{f.glitzner.1@research.gla.ac.uk}}
\date{}

\maketitle

\begin{abstract}
    Stable matching is a fundamental area with many practical applications, such as centralised clearinghouses for school choice or job markets. Recent work has introduced the paradigm of near-feasibility in capacitated matching settings, where agent capacities are slightly modified to ensure the existence of desirable outcomes. While useful when no stable matching exists, or some agents are left unmatched, it has not previously been investigated whether near-feasible stable matchings satisfy desirable properties with regard to their stability in the original instance. Furthermore, prior works leave open deviation incentive issues that arise when the centralised authority modifies agents' capacities.
    
    We consider these issues in the {\sc Stable Fixtures} problem model, which generalises many classical models through non-bipartite preferences and capacitated agents. We develop a formal framework to analyse and quantify agent incentives to adhere to computed matchings. Then, we embed near-feasible stable matchings in this framework and study the trade-offs between instability, capacity modifications, and computational complexity. We prove that capacity modifications can be simultaneously optimal at individual and aggregate levels, and provide efficient algorithms to compute them. We show that different modification strategies significantly affect stability, and establish that minimal modifications and minimal deviation incentives are compatible and efficiently computable under general conditions. Finally, we provide exact algorithms and experimental results for tractable and intractable versions of these problems.
\end{abstract}


\section{Introduction}
\label{sec:intro}

Stable matching theory lies at the heart of centralised matching and allocation mechanisms for multi-agent systems \cite{matchup}. The classical {\sc Stable Marriage} problem deals with two equally-sized and non-overlapping sets of agents (which we will refer to as a \emph{bipartition}), which are traditionally referred to as men and women (but could also be students and projects, for example), and their strict ordinal preference relations over agents from the set that they are not a part of (i.e., their potential partners). A solution to the problem is a matching (a pairing among the agents) such that there does not exist a \emph{blocking pair}, i.e., a pair of agents that strictly prefer each other to their assigned partners (either or both of which might be none). This is important, as a blocking pair gives agents an incentive for decentralised coordination and deviation, which might cause unravelling. We refer to a matching free of blocking pairs as \emph{stable}, and as \emph{unstable} otherwise. Gale and Shapley showed that, in this model, a stable matching always exists and can be found efficiently using the celebrated Deferred Acceptance mechanism \cite{gale_shapley}.

Two important extensions of the {\sc Stable Marriage} problem involve the existence of \emph{capacitated agents}, i.e., agents with capacities higher than 1 (where the capacity indicates an upper limit on the number of matches that the agent can be a part of), and the possibility of \emph{non-bipartite preferences}, i.e., the absence of a bipartition (for example, men and women or students and projects). One problem model that permits both of these extensions simultaneously is given by the {\sc Stable Fixtures} problem ({\sc sf}). This general many-to-many matching model captures a wide range of real-world multi-agent interactions: from assigning students to projects and residents to hospitals, to pairing competitors in tournaments and participants in kidney exchanges \cite{matchup}. However, unlike bipartite matching settings, these environments often lack desirable structural guarantees: stable outcomes may not exist \cite{gusfield89}. This motivates the search for alternative solution concepts that maintain minimal deviation incentives and desirable stability properties in practice.

Two prevalent directions for alternative solutions are \emph{almost-stability} and \emph{near-feasibility}, where the aim with the former is to find a matching with a minimum number of blocking pairs (or a minimum number of blocking agents or a minimax number of blocking pairs per agent) \cite{abraham06,biro_sm_10,chen17,glitznermanloveasaamas26}, to minimise the likelihood of the matching being undermined \cite{Eriksson2008,matchup}. The aim with near-feasibility is to find a new instance in which few agents' capacities are modified in such a way that the resulting new instance guarantees the existence of one or more stable matchings \cite{nguyen18}. From a multi-agent perspective, these notions reflect two distinct approaches to managing instability: either tolerate limited incentives to deviate (almost-stability), or minimally adjust system parameters (capacities) to restore collective stability (near-feasibility).

Unfortunately, most problems involving the computation of almost-stable matchings do not admit efficient algorithms (under standard complexity assumptions), even in very restricted settings \cite{abraham06,biro10,biro12,chen17,chen2019computational}. On the other hand, there is a very active line of research that investigates the tractability frontier and the design of efficient algorithms for problems involving notions of near-feasibility  \cite{chencsaji23,couples24,gokhale24,ranjan25,controlsurvey25,cembrano25,gergely25,glitzner2025unsolvabilitymanytomanynonbipartitestable,nguyen19,nguyen21,glitzner2024structuralteac} and the related problem of \emph{capacity planning} \cite{bobbio23,bazotte25,afacan24}. However, one major shortcoming of near-feasible stable matchings is that they do not necessarily have desirable stability properties with respect to the original instance (without modified capacities). In particular, in a near-feasible stable matching, agents might end up in strictly fewer or strictly more pairs than desired, leading to an inherent instability. Furthermore, many current algorithms do not differentiate between the agents whose capacity they modify and agents whose capacity they do not modify, whereas it would seem more natural to change the capacity of an agent with high capacity rather than one with low capacity due to the marginally less significant impact on the difference in their number of matches, for example.

\subsection{Contributions and Structure}

In this paper, we study the incentive and complexity landscape underlying near-feasible stable matchings in multi-agent systems where agents can form multiple partnerships. Within the {\sc Stable Fixtures} model, we investigate the trade-off between system-level feasibility and individual-level stability, and develop new methods to quantify and manage agent incentives under capacity modifications. Our main contributions are fourfold:
\begin{enumerate}
    \item We introduce new concepts and problem models that connect near-feasibility with explicit agent-level incentive measures.
    \item We develop efficient algorithms for finding matchings and capacity modifications that balance feasibility and stability.
    \item We derive tight bounds and structural characterisations for these settings, highlighting the interplay between incentives and modifications at the individual and aggregate levels.
    \item We show through an empirical investigation that the minimum number of capacity modifications required to arrive at solvability and the minimum instability are often small.
\end{enumerate}
Beyond theoretical interest, our results contribute to the broader goal of designing fair, desirable, and scalable coordination mechanisms for multi-agent environments.

The remainder of the paper develops these ideas in more detail. Section \ref{sec:model} introduces the model and highlights related work. In Section \ref{sec:mod}, we systematically study different kinds of capacity modifications that lead to solvable instances. We introduce natural optimality notions that aim to minimise the effect of the modifications either on the individual or on the aggregate level, and study the structure and computation of these optimal modifications when allowing just downwards, just upwards, or either kind of changes. Crucially, for each of these six settings, we present simple algorithms that compute such a modification efficiently, and we give tight bounds on the individual and total number of changes required. Furthermore, we show that a well-balanced capacity modification (leading to a solvable instance) always exists. We also show that there is always common ground between optimal capacity modifications on the individual and aggregate level, and highlight which efficient algorithms satisfy these criteria simultaneously.

Next, in Section \ref{sec:instability}, we consider the incentives that agents may experience to deviate from these previously studied modification-optimal near-feasible stable matchings. We extend the classical concept of \emph{blocking pairs} to quantify instability to a more applicable notion of \emph{blocking entries}, which represent non-symmetric agent incentives to deviate, and define optimality criteria for matchings on the individual and aggregate levels. We study the computation of optimal matchings in a number of settings that impose different restrictions on the allowed capacity changes, establishing the tractability frontier and presenting bounds on the instability of different kinds of matchings along the way. While the computation of matchings that minimise deviation incentives without increasing capacities turns out to be {\sf NP-hard}, we find that, surprisingly, matchings that minimise the deviation incentives on the individual and aggregate levels simultaneously, while never requiring capacity changes stronger than an optimal modification, can be computed efficiently. 

Section \ref{sec:exact} contains exact algorithms and establishes {\sf XP} results for {\sf NP-hard} versions of these problems. In Section \ref{sec:exp}, we implement some of these algorithms and study optimal capacity modifications and instability-minimising matchings in practice using synthetic datasets. Finally, Section \ref{sec:conclusion} concludes the paper and presents some open problems.

\section{Background}
\label{sec:model}

To formally reason about self-interested agents with limited capacities, we adopt the {\sc Stable Fixtures} model (which we previously introduced informally in Section \ref{sec:intro}). This framework generalises classical stable matching models to settings where each agent may form multiple simultaneous partnerships. A formal definition follows.

\begin{definition}[{\sc sf} Instance]
    Let $I=(A,\succ, c)$ be an {\sc sf} \emph{instance} where $A=\{ a_1, a_2, \dots, a_n \}$, also denoted by $A(I)$, is a set of $n\in \mathbb{N}$ agents, $\succ$ is a tuple of $n$ strict ordinal \emph{preference rankings} $\succ_i$ for each agent $a_i$ over all other agents $A\setminus\{a_i\}$, and $c:A\rightarrow \{1,2,\dots,n-1\}$ is a capacity function indicating the maximum number $c_i$ of agents that each agent $a_i\in A$ can be matched to.
\end{definition}

Intuitively, each agent $a_i$ can form up to $c_i$ bilateral relationships with other agents, reflecting real-world multi-agent interactions such as shared projects or cooperative tasks. We can equivalently treat an instance $I=(A,\succ,c)$ as a graph $G=(A,E)$: Let every agent $a_i\in A$ be a vertex in our graph and let all pairs of agents $\{a_i,a_j\}$ in $\succ$ form an edge $e$ in the edge set $E$. Then any {\sc sf} instance $I$ induces an \emph{acceptability graph} isomorphic to the complete graph on $n$ vertices. Throughout this paper, we generally assume \emph{complete} and \emph{truthful} preferences, i.e., that every pair of agents is considered acceptable (it has been shown that results generally carry over to the setting with \emph{incomplete preference lists} \cite{gusfield89}) and that the problem instance contains preference rankings and capacities that represent the true preferences and desired capacities of the agents.\footnote{This is a strong but widely adopted assumption in this area, as it is well-known that no mechanism that computes a stable solution is \emph{strategyproof} even in the restricted {\sc Stable Marriage} model \cite{roth82}. We note, though, that Bazotte, Carvalho and Vidal \cite{bazotte25} recently studied the effects of strategic agents on capacity planning in matching schemes.} We also implicitly assume that every agent would rather be matched to any other agent than to be left with free capacity. The stability notion that we adopt throughout this paper was introduced by \citet{Irving2007}, and is a natural extension of classical stability in the {\sc Stable Marriage} model. It is defined as follows.

\begin{definition}[Stable Matching]
    Let $I=(A,\succ, c)$ be an {\sc sf} instance. A \emph{matching} $M$ of $I$ is a set of unordered pairs of agents such that no agent $a_i$ is contained in more than $c_i$ pairs.\footnote{Sometimes a matching with such capacity constraints is referred to as a \emph{b-matching} to indicate a capacity function $b$ ($c$ in our notation).} Let $M(a_i)$ denote the set of agents assigned to $a_i$ and let worst$_i(M(a_i))$ denote the worst agent assigned to $a_i$ in $M$ (according to $\succ_i$). Then a \emph{blocking pair} of $M$ is a pair of distinct agents $a_i, a_j\in A$ such that 
    \begin{enumerate}[label=(\roman*),leftmargin=2.1em]
        \item $\{a_i,a_j\}\notin M$, and
        \item $\vert M(a_i)\vert<c_i$ or $a_j \succ_i$ worst$_i(M(a_i))$, and
        \item $\vert M(a_j)\vert< c_j$ or $a_i\succ_j$ worst$_j(M(a_j))$. 
    \end{enumerate}
    Let $bp(M)$ denote the set of blocking pairs admitted by $M$. If $bp(M)=\varnothing$, then $M$ is referred to as \emph{stable}.
\end{definition}

In multi-agent terms, stability corresponds to a state in which no two agents can jointly improve their outcomes by changing their current associations. We refer to an {\sc sf} instance $I=(A,\succ,c)$ that admits at least one stable matching as \emph{solvable}; otherwise, we refer to $I$ as \emph{unsolvable}. We refer to a matching $M$ as a \emph{near-feasible matching} of $I$ if there exists a capacity function $c'$ such that $M$ is a matching in $I'=(A,\succ,c')$. Furthermore, $M$ is a \emph{near-feasible stable matching} of $I$ if $M$ is stable in $I'$. Finally, we refer to $I'$ as a \emph{near-feasible instance}.

\subsection{Motivating Example}

Beyond other important problems, such as the assignment of doctors to hospitals, pupils to schools, students to colleges, etc., that appear as special cases of our model, we provide the following, arguably less pressing, motivating example to illustrate the full potential of the {\sc sf} model. 

Consider a group of academics who are looking for new partners to start a (pairwise) research collaboration with. Suppose, furthermore and for the sake of example, that this is facilitated through a central match-making service in which academics sign up, state their capacity on the maximum number of new research partners, and rank these potential partners in decreasing order of preference (e.g., based on preference factors such as geographic distance, research expertise, similarity of profile). It is then the task of the match-making service to find a stable matching among the academics: a matching in which every participant's capacity is respected and no two academics that are not yet matched to each other have free capacity or would like to drop someone else that they are matched to in order to enter this new collaboration instead. If a stable matching does not exist, then we can try to relax the capacity constraints of academics and look for a stable matching under this new set of constraints. This is what we would refer to as a near-feasible stable matching.

\begin{example}
In Table \ref{table:solvablesf}, academic $a_1$ has capacity 2 and ranks the academics $a_2$, $a_3$, $a_4$, and $a_5$ in decreasing order of preference (and similarly for the remaining academics). The problem instance admits the stable matching $M=\{\{a_1, a_2\}, \{a_1, a_3\}, \{a_2, a_3\}, \{a_4, a_5\}\}$, in which academics $a_1,a_2$ and $a_3$ are fully matched (i.e., their number of assigned matches equals their capacity) and academics $a_4$ and $a_5$ each remain with one unit of free capacity. On the other hand, Table \ref{table:unsolvablesf} shows an unsolvable {\sc sf} instance (with a generalised stable partition (GSP \cite{glitzner2025unsolvabilitymanytomanynonbipartitestable}) $\Pi=(a_1\;a_2\;a_3)(a_1\;a_4)(a_2\;a_4)(a_3\;a_5)$ indicated in dotted and unbroken circles). Notice that in Table \ref{table:unsolvablesf}, if one of the academics $a_1$, $a_2$ or $a_3$ were to increase or to decrease their capacity by 1, then the resulting instance would be solvable -- we would consider this to be a near-feasible instance $I'$ and a corresponding stable matching $M'$ of $I'$ to be a near-feasible stable matching in $I$.

\begin{table}[!htb]
\centering
\begin{minipage}{.45\linewidth}
\centering
\caption{A solvable {\sc sf} instance}
    \begin{tabular}{ c | c | c c c c }
    $a_i$ & $c_i$ & pref \\\hline
    $a_1$ & 2 & $\boxed{a_2}$ & $\boxed{a_3}$ & $a_4$ & $a_5$ \\
    $a_2$ & 2 & $\boxed{a_1}$ & $\boxed{a_3}$ & $a_5$ & $a_4$ \\
    $a_3$ & 2 & $\boxed{a_1}$ & $\boxed{a_2}$ & $a_4$ & $a_5$ \\
    $a_4$ & 2 & $a_1$ & $a_2$ & $a_3$ & $\boxed{a_5}$ \\ 
    $a_5$ & 2 & $a_2$ & $a_1$ & $a_3$ & $\boxed{a_4}$ 
    \end{tabular}
\label{table:solvablesf}
\end{minipage}%
\begin{minipage}{.45\linewidth}
\centering
\caption{An unsolvable {\sc sf} instance}
    \begin{tabular}{ c | c | c c c c }
    $a_i$ & $c_i$ & pref \\\hline
    $a_1$ & 2 & \circled{$a_2$} & $\boxed{a_4}$ & \dottedcircled{$a_3$} & $a_5$ \\
    $a_2$ & 2 & $\boxed{a_4}$ & \circled{$a_3$} & \dottedcircled{$a_1$} & $a_5$ \\
    $a_3$ & 2 & \circled{$a_1$} & $\boxed{a_5}$ & \dottedcircled{$a_2$} & $a_4$ \\
    $a_4$ & 2 & $a_3$ & $\boxed{a_1}$ & $\boxed{a_2}$ & $a_5$ \\
    $a_5$ & 1 & $\boxed{a_3}$ & $a_1$ & $a_2$ & $a_4$ 
    \end{tabular}
\label{table:unsolvablesf}
\end{minipage} 
\end{table}
\end{example}

We will provide results for this setting, but also develop ideas and techniques that are applicable more widely.

\subsection{Related Work}
\label{sec:relatedwork}

Research on stable matchings and their generalisations spans several communities, including multi-agent systems, computational social choice, and algorithmic game theory \cite{matchup}. Our work contributes to this literature by connecting two major directions -- capacity modification and instability minimisation -- within the general {\sc sf} framework.

{\sc sf} was introduced by Irving and Scott \cite{Irving2007} as a natural many-to-many extension of the one-to-one {\sc Stable Roommates} problem ({\sc sr}). {\sc sr} is, of course, a very well-studied problem.\footnote{See, for example, the monographs by Gusfield and Irving \cite{gusfield89} and Manlove \cite{matchup}, early papers such as \cite{irving_sr_structure,irving_sr,gusfield_sr_structure,ronn_srt}, or more recent work in a range of domains encompassing economic properties, experimental studies and algorithm design, for example \cite{Chen2025,csehshort,glitzner2025empirics,glitzner2024structuralteac,herings25,mertens15random}.} Irving and Scott gave an algorithm to find a stable matching for an {\sc sf} instance in linear time, if one exists \cite{Irving2007}. Interesting phenomena that only occur in {\sc sf} but not in its one-to-one counterpart were illustrated by \citet{glitzner2025unsolvabilitymanytomanynonbipartitestable} and \citet{matchup}. \citet{birocsaji24} proved that even under lexicographic preferences, the strong core of an {\sc sf} instance can be non-empty, and finding weak or strong core matchings (if they exist) is {\sf NP-hard}. Other many-to-many generalisations of {\sc sr} focus on additional edge capacities \cite{BiroPhD08, biro10,fleinerphd18} or parallel edges \cite{sma05,cechlarova05,bmatchingrotations}, for example.

Most related work on capacity modifications is for more restricted bipartite many-to-one matching problem models, such as versions of the {\sc Hospitals/Residents} problem ({\sc hr}) \cite{nguyen18,nguyen19,nguyen21,couples24,chencsaji23,gokhale24}. {\sc hr} appears as a special case of {\sc sf} when preferences are bipartite and incomplete, and one set of agents in the bipartition has capacity 1. This problem is also known as {\sc College Admission} or {\sc School Choice} to reflect other contexts in which this model is applicable \cite{afacan24}. Such previous works on the {\sc hr} problem that are similar to ours usually only consider hospital capacity increases and not decreases, as it is of practical importance that as many doctors as possible are assigned a position. 

Nguyen and Vohra \cite{nguyen18} initiated the study of algorithms to find near-feasible stable matchings and focused on an extension of {\sc hr} in which doctors can form couples and submit joint preferences. Contrary to classical {\sc hr}, in this problem variant, the existence of a stable matching is not guaranteed. However, the authors showed that hospital capacities can be perturbed such that the total aggregate capacity is never reduced and can increase by at most four, the capacity of each hospital never changes by more than two, and the resulting instance admits a stable matching. On the technical side, they use Scarf's Lemma \cite{scarf} to find a fractional stable matching and then iteratively round the matching and change the capacities correspondingly. Follow-up work considered more complex extensions of the problem \cite{nguyen19,nguyen21}. 

The work on {\sc hr} with couples was extended recently by \citet{couples24}. They highlight that the problem of computing a stable fractional matching is {\sf PPAD-hard} but needs to be solved in the method by \citet{nguyen18}. Therefore, instead, \citet{couples24} focus on a family of special cases of the problem and provide a polynomial-time algorithm that adjusts hospital capacities by at most 1, which is both a conceptual and an algorithmic improvement compared to \citet{nguyen18} in these special cases. They also provide hardness results for various problems, even under strong assumptions on the instances, and mention that an important future direction is to continue to expand our knowledge of the frontier between tractable and intractable variants of near-feasible stable matching, which we continue in this paper. A different extension of the results by \citet{nguyen18} was recently presented by \citet{cembrano25}. They establish a similar new rounding theorem and apply it to a range of problems, including {\sc hr} with couples. 

Chen and Cs{\'a}ji \cite{chencsaji23} also considered the {\sc hr} problem and studied how to optimally increase the capacities of hospitals to obtain desirable stable matchings. Afacan, Dur and Van der Linden \cite{afacan24} considered similar problems, but under fundamentally different assumptions. Gokhale et al. \cite{gokhale24} also studied {\sc hr} and investigated, for example, how the set of stable matchings changes when capacities are increased. They pose an open problem of exploring algorithms for capacity modification when both increase and decrease operations are allowed, which we pursue in this paper.

Ranjan, Nasre and Nimbhorkar \cite{ranjan25} investigated {\sc hr} with ties under strong-stability\footnote{See \citet{matchup} for a discussion of stability definitions in the presence of ties.} and presented special cases where finding increased capacities that yield a solvable instance such that the total sum of capacity increases is minimal is polynomial-time solvable. Minimising the maximum capacity increase, however, turned out to be {\sf NP-hard} even under strong assumptions on the preferences. Bobbio et al. \cite{bobbio23} and Bazotte, Carvalho and Vidal \cite{bazotte25} considered optimisation problems related to capacity planning in {\sc hr}, and presented an operations research perspective for fairness considerations and strategic issues.

A different but related line of work studies control and manipulation problems, initiated by Boehmer et al. \cite{boehmer21}. Chen and Schlotter \cite{Chen2025} have since extended these results to {\sc sr}. For our model, B{\'e}rczi, Cs{\'a}ji and Kir{\'a}ly \cite{manipulation24} established that when removing agents rather than modifying capacities, finding a subset of agents to remove of smallest size to arrive at a solvable instance is {\sf NP-hard}. Chen et al. \cite{controlsurvey25} recently surveyed relevant results for control problems in the area of computational social choice more broadly.

In a recent preprint, Cs{\'a}ji \cite{gergely25} studied capacity changes in very general hypergraph matching problems. For {\sc sf}, he presented a polynomial-time algorithm that changes capacities by at most one and such that the total sum of capacity changes is either 0 or +1. Minimality of the amount of changes was not established. In a similar line of work, Glitzner and Manlove \cite{glitzner2025unsolvabilitymanytomanynonbipartitestable} provided structural and algorithmic results for {\sc sf}. They also gave a new polynomial-time algorithm with the same structural guarantees that exploits inherent structures of the preference system, which is fundamentally different from the algorithm by Cs{\'a}ji \cite{gergely25} (which uses Scarf's Lemma \cite{scarf}). Again, no claims about the minimality of changes were made, which is something we will resolve here.

With regard to minimising the instability of matchings, related work mainly focuses on one-to-one matching models. As a natural way to deal with unsolvable {\sc sr} instances, Abraham, Bir{\'o} and Manlove \cite{abraham06} introduced the problem of finding \emph{almost-stable} matchings, which are matchings with the minimum number of blocking pairs. Finding such a matching is known to be {\sf NP-hard} even for preference lists of length at most 3 \cite{biro12}, {\sf W[1]-hard} with respect to the optimal value \cite{chen17}, and {\sf NP-hard} to approximate within any constant factor \cite{abraham06}. Chen et al. \cite{chen17} showed that the problem of minimising the number of unique agents in blocking pairs is also {\sf W[1]-hard} (with respect to the optimal value). Glitzner and Manlove recently showed that even deciding whether there exists a matching in which no agent is in more than one blocking pair is {\sf NP-complete} \cite{glitznermanloveminimax}, and deciding the existence of a matching in which a given subset of agents is in no blocking pairs is also {\sf NP-complete} \cite{glitznermanlovepref}. Almost-stable matchings have also been investigated for other problem models (e.g., see \cite{biro_sm_10,chen2019computational,couples24,gupta2020parameterized,hamada09,minbp_hrc_17}). Alternative solution concepts for unsolvable instances that are fundamentally different from almost-stability have also been explored (e.g., see \cite{tan91_2,biro16,herings25,vandomme2025locally,ATAY2021102465}).

\section{Minimal Capacity Modifications}
\label{sec:mod}

In this section, we first propose and formalise notions of optimality for modified capacity functions and then show how to compute such optimal functions. We are interested in how minimal interventions in agent capacities can restore stability in otherwise unsolvable instances.

\subsection{Preliminary and Known Results}

Recall from the introduction that we are interested in changing the capacity function of unsolvable {\sc sf} instances by small amounts to make these new resulting instances solvable. To specify more formally what is meant by ``small amounts'', we begin by defining an optimality concept that imposes restrictions on the capacity modifications of \emph{individual} agents.

\begin{definition}
    A \emph{minimal individual modification} (denoted \emph{MIM}$_{\pm}$) is an alternative capacity function $c'$ of an unsolvable {\sc sf} instance $I=(A,\succ,c)$ such that $I'=(A,\succ,c')$ is solvable and there does not exist an alternative capacity function $c''$ such that $I''=(A,\succ,c'')$ is solvable and $\max_{a_i\in A}\vert c_i''-c_i\vert<\max_{a_i\in A} \vert c_i'-c_i\vert$.
    
    Furthermore, $c'$ is \emph{MIM}$_\text{+}$ if, for all agents $a_i\in A$, $c_i'\geq c_i$, $I'=(A,\succ,c')$ is solvable, and there does not exist an alternative capacity function $c''$ such that $I''=(A,\succ,c'')$ is solvable, $c_i''\geq c_i$ and $\max_{a_i\in A}(c_i''-c_i)<\max_{a_i\in A} (c_i'-c_i)$.

    Similarly, $c'$ is \emph{MIM}$_{-}$ if, for all agents $a_i\in A$, $c_i'\leq c_i$, $I'=(A,\succ,c')$ is solvable, and there does not exist an alternative capacity function $c''$ such that $I''=(A,\succ,c'')$ is solvable, $c_i''\leq c_i$ and $\max_{a_i\in A}(c_i-c_i'')<\max_{a_i\in A} (c_i-c_i')$.
\end{definition}

Informally, MIM focuses on minimising the largest adjustment of any agent. An alternative capacity function $c'$ is MIM$_{\pm}$ if it makes the new instance solvable and requires the least extreme individual capacity changes among all alternative capacity functions that make the instance solvable. MIM$_\text{+}$ (MIM$_{-}$) requires, additionally, that capacities are only ever modified upwards (downwards) and that the new capacity function requires the least extreme individual upwards (downwards) changes among all such alternative capacity functions. 

Alternatively, we will consider optimality at the \emph{aggregate} level, i.e., with respect to the total number of capacity modifications among all agents, again separating the case in which any capacity modifications are allowed and cases where only upwards or only downwards modifications are allowed. Formal definitions of the optimality concepts follow.

\begin{definition}
    A \emph{minimal aggregate modification} (denoted \emph{MAM}$_{\pm}$) is an alternative capacity function $c'$ of an unsolvable {\sc sf} instance $I=(A,\succ,c)$ such that $I'=(A,\succ,c')$ is solvable and there does not exist an alternative capacity function $c''$ such that $I''=(A,\succ,c'')$ is solvable and $\sum_{a_i\in A}\vert c_i''-c_i\vert<\sum_{a_i\in A} \vert c_i'-c_i\vert$.
    
    Furthermore, $c'$ is \emph{MAM}$_\text{+}$ if, for all agents $a_i\in A$, $c_i'\geq c_i$, $I'=(A,\succ,c')$ is solvable, and there does not exist an alternative capacity function $c''$ such that $I''=(A,\succ,c'')$ is solvable, $c_i''\geq c_i$ and $\sum_{a_i\in A}(c_i''-c_i)<\sum_{a_i\in A} (c_i'-c_i)$.

    Similarly, $c'$ is \emph{MAM}$_{-}$ if, for all agents $a_i\in A$, $c_i'\leq c_i$, $I'=(A,\succ,c')$ is solvable, and there does not exist an alternative capacity function $c''$ such that $I''=(A,\succ,c'')$ is solvable, $c_i''\leq c_i$ and $\sum_{a_i\in A}(c_i-c_i'')<\sum_{a_i\in A} (c_i-c_i')$.
\end{definition}

Informally, MAM aims to minimise the total adjustments across all agents.

Now, to investigate the structure and computation of these different kinds of optimal capacity modifications, it will be useful to consider the \emph{generalised stable partition} (GSP) structure\footnote{Two equivalent definitions were introduced; we adopt that of a GSP1.} introduced by Glitzner and Manlove \cite{glitzner2025unsolvabilitymanytomanynonbipartitestable}, which builds on the well-established stable partition structure first introduced by Tan \citet{tan91_2,tan91_1}, as it succinctly captures relevant structures of an {\sc sf} instance in the form of a collection of agent permutations.\footnote{For the following definition of a GSP, a \emph{cyclic permutation} (also referred to as a \emph{$k$-perm} for a cyclic permutation with cycle length $k$) is a permutation consisting of a single cycle. For example, for a set $S=\{x,y,z\}$, we would consider the permutations $(x\;y\;z)$ and $(x\;z\;y)$ to be cyclic permutations of $S$ because they consist of a single cycle. We consider two permutations to be \emph{distinct} when there exists an element that is mapped to two different elements in the respective permutations. As pointed out in the definition, fixed points need not be distinct, e.g., $(a_1)(a_1)$ represents two different occurrences of $a_1$ in a fixed point.}

\begin{definition}[GSP]
\label{def:gsp1}
    Let $I=(A,\succ, c)$ be an {\sc sf} instance. Then a \emph{GSP} $\Pi$ is a collection of $k$ distinct (apart from fixed points) cyclic permutations $\Pi_i$ (for $1\leq i\leq k$) of sets $A_i\subseteq A$ such that  
    \begin{enumerate}[leftmargin=2.6em]
        \item[(F1)] $\forall \; \Pi_i\in \Pi$ and $\forall \; a_j\in A_i$ we have $\Pi_i(a_j)\succeq_j \Pi_i^{-1}(a_j)$, 
        \item[(F2)] $\nexists \; a_i, a_j \in A$ where $a_i\neq a_j$ and $(a_i\; a_j)\notin \Pi$ such that $a_j \succ_i \Pi_r^{-1}(a_i)$ and $a_i \succ_j \Pi_s^{-1}(a_j)$ for some $\Pi_r, \Pi_s\in \Pi$,
        \item[(F3)] $\forall a_i\in A$ we have $\vert \{r \; \vert \; a_i\in A_r\}\vert = c_i$, and
        \item[(F4)] $\forall a_i, a_j \in A$ with $a_i\neq a_j$ we have that $\vert \{s\in \{1,2,\dots,k\}\; \vert \; \Pi_s(a_i) = a_j\} \vert + \vert \{s \; \vert \; \Pi_s(a_j) = a_i\} \vert \leq 2$.
    \end{enumerate}
    We call a GSP \emph{reduced} if it does not contain any cycles of even length longer than 2.
\end{definition}

The following are examples of GSPs, written as collections of permutations in cyclic notation, of the instances shown in Tables \ref{table:solvablesf}-\ref{table:unsolvablesf}, respectively: $\Pi=(a_1\;a_2)(a_1\;a_3)(a_2\;a_3)(a_4\;a_5)$ and $\Pi=(a_1\;a_2\;a_3)(a_1\;a_4)(a_2\;a_4)(a_3\;a_5)$. \citet{glitzner2025unsolvabilitymanytomanynonbipartitestable} established the following crucial properties of GSPs.

\begin{theorem}[\cite{glitzner2025unsolvabilitymanytomanynonbipartitestable}]
\label{thm:gsp}
    Let $I$ be an {\sc sf} instance with $n$ agents, then 
    \begin{itemize}[leftmargin=1.5em]
        \item $I$ admits at least one and possibly exponentially many GSPs;
        \item A reduced GSP of $I$ can be computed in $O(n^3)$ time;
        \item Any two GSPs of $I$ admit the same cycles of odd length;
        \item The collection of transpositions of a GSP that does not admit cycles of length longer than 2 corresponds to a stable matching in $I$;
    \end{itemize}
\end{theorem}

Using these properties and some further technical results, they presented an algorithm for {\sc sf} instances $I=(A,\succ,c)$ to find a capacity modification $c'$ and a stable matching $M'$ of $I'=(A,\succ,c')$ such that each agent's capacity is changed by at most 1 and the total sum of capacity changes is either 0 or +1. It works as follows: start with a reduced GSP $\Pi$ of $I$ and let $\mathcal{O}_I^{\geq 3}$ denote the cycles of $\Pi$ of odd length at least 3. Then pick exactly one agent $a_r$ from every cycle in $\mathcal{O}_I^{\geq 3}$ (it was shown that no agent can be contained in more than one such cycle), and alternately increase and decrease these $a_r$ agents' capacities by 1. Due to the alternating increases and decreases, we will denote this algorithm by \texttt{NearFeasible$_\pm$}. The algorithm was shown to run in $O(n^2)$ time (where $n$ is the number of agents) given a GSP. Note that the algorithm assumes that a GSP is given, so the stated complexity increases from $O(n^2)$ to $O(n^3)$ if a GSP has not yet been computed (when using the method from \cite{glitzner2025unsolvabilitymanytomanynonbipartitestable} to compute a GSP). Note, though, that an {\sc sf} instance has size $\Theta(n^2)$ (by our assumption of complete preferences), so a runtime of $O(n^3)$ is still sub-quadratic in the input size. The same remark applies to our stated complexities throughout.

By Lemma 3 of \cite{glitzner25sagt}, though, it is immediate that it suffices to pick any one agent from each cycle of odd length at least 3, and that the algorithm works regardless of whether these agents' capacities are increased or decreased by 1, as long as each such agent's capacity is modified by exactly 1. Let \texttt{NearFeasible$_\text{+}$} and \texttt{NearFeasible$_-$} denote versions of the algorithm that only ever increase or only ever decrease capacities by 1, respectively. In this paper, we will put this resulting suite of algorithms to the test and develop the necessary frameworks and combinatorial results to prove optimality from a variety of angles. For completeness, we give the full pseudocode of these algorithms in Appendix \ref{sec:appendixAlgos}.

\subsection{Computing Optimal Modifications}

By changing everyone's capacity by at most 1, we can immediately observe the following.

\begin{prop}
\label{prop:mim}
    Let $I=(A,\succ,c)$ be an unsolvable {\sc sf} instance with $n$ agents. Then
    \begin{itemize}[leftmargin=1.5em]
        \item a \emph{MIM$_\text{+}$} capacity function $c'$ always has $\max_{a_i\in A} (c_i'-c_i)=1$ and \emph{\texttt{NearFeasible$_\text{+}$}} computes such a solution efficiently;
        \item a \emph{MIM$_{-}$} capacity function $c'$ always has $\max_{a_i\in A} (c_i-c_i')=1$ and \emph{\texttt{NearFeasible$_-$}} computes such a solution efficiently;
        \item any \emph{MIM$_\text{+}$} and any \emph{MIM$_{-}$} capacity function $c'$ is a \emph{MIM$_{\pm}$} capacity function with $\max_{a_i\in A} \vert c_i'-c_i\vert=1$ and thus a \emph{MIM$_{\pm}$} capacity function can be computed efficiently.
    \end{itemize}
\end{prop}

In fact, conceptually, this result shows that we can always find well-balanced capacity modifications in the sense that, in absolute numbers, no single agent's capacity modification is significantly higher than any other's.

Optimality of the algorithms on the aggregate level (MAM), on the other hand, is not immediate. However, surprisingly, it turns out that here, too, these algorithms perform optimally.

\begin{theorem}
\label{thm:mam}
    Let $I=(A,\succ,c)$ be an unsolvable {\sc sf} instance with $n$ agents and let $\Pi$ be a GSP of $I$. Furthermore, let $\mathcal{O}_I^{\geq3}$ be the cycles of odd length at least 3 of $\Pi$. Then
    \begin{itemize}[leftmargin=1.5em]
        \item a \emph{MAM$_\text{+}$} capacity function $c'$ always has $\sum_{a_i\in A} (c_i'-c_i)=\vert\mathcal{O}_I^{\geq3}\vert$ and such a $c'$ can be computed efficiently using \emph{\texttt{NearFeasible$_\text{+}$}};
        \item a \emph{MAM$_{-}$} capacity function $c'$ always has $\sum_{a_i\in A} (c_i-c_i')=\vert\mathcal{O}_I^{\geq3}\vert$ and such a $c'$ can be computed efficiently using \emph{\texttt{NearFeasible$_-$}};
        \item any \emph{MAM$_\text{+}$} and any \emph{MAM$_{-}$} capacity function $c'$ is a \emph{MAM$_{\pm}$} capacity function with $\sum_{a_i\in A} \vert c_i'-c_i\vert=\vert\mathcal{O}_I^{\geq3}\vert$ and thus a \emph{MAM$_{\pm}$} capacity function can be computed efficiently.
    \end{itemize}
\end{theorem}
\begin{proof}
    The first observation to make is that the odd-length cycles are invariant among all GSPs by Theorem 3.4. Then it suffices to show that even when allowing capacity modification in any direction and of any intensity to arrive at an alternative capacity function $c'$ such that $I'=(A,\succ,c')$ is solvable, it must be the case that $\sum_{a_i\in A} \vert c_i'-c_i\vert\geq\vert\mathcal{O}_I^{\geq3}\vert$. Assuming this key property, and by the fact that all of \texttt{NearFeasible$_\pm$}, \texttt{NearFeasible$_+$} and \texttt{NearFeasible$_-$} require exactly $\vert\mathcal{O}_I^{\geq3}\vert$ changes, i.e., they all achieve $\sum_{a_i\in A} \vert c_i'-c_i\vert=\vert\mathcal{O}_I^{\geq3}\vert$ and $I'=(A,\succ,c')$ is solvable, the correctness of the algorithms for computing MAM$_{\pm}$, MAM$_{+}$ and MAM$_{-}$ capacity functions, respectively, becomes immediate.

    We will prove the key property by induction: that is, we will show that for any $c'$ and any $x\in \{1,2,\dots,\sum_{a_i\in A} c_i\}$ such that $\sum_{a_i\in A} \vert c_i'-c_i\vert=x$, it must be the case that $\vert\mathcal{O}_{I'}^{\geq 3}\vert\geq \vert\mathcal{O}_I^{\geq 3}\vert-x$ (we will refer to this statement as $P(x)$), i.e., that $x$ amount of changes to the capacity function decrease the number of odd cycles in the preference system by at most $x$. 
    
    Consider the base case where $x=1$, i.e., pick one agent $a_i$ whose capacity is changed either upwards or downwards by 1. Let $\Pi$ be a reduced GSP of $I$. 
    \begin{enumerate}
        \item First, suppose that $a_i$ is in a cycle $C=(a_i\;a_{r_1}\; a_{r_2} \dots a_{r_k})$ of odd length at least 3 in $\Pi$. Then if $a_i$'s capacity is increased, $\Pi'=(\Pi\setminus C)\cup (a_i\;a_{r_1})(a_{r_2}\;a_{r_3})\dots(a_{r_{k}}\;a_{i})$ is a GSP of $I'$. If $a_i$'s capacity is decreased instead, then $\Pi'=(\Pi\setminus C)\cup (a_{r_1}\; a_{r_2})(a_{r_3}\;a_{r_4})\dots(a_{r_{k-1}}\;a_{k})$ is a GSP of $I'$ (both of these statements follow directly from Lemma 3 of \cite{glitzner25sagt}). Thus, in both cases, $\vert\mathcal{O}_{I'}^{\geq 3}\vert= \vert\mathcal{O}_I^{\geq 3}\vert-1$. 
        
        \item Suppose instead that $a_i$ is not in a cycle of odd length at least 3 in $\Pi$ and that $c_i'=c_i+1$. Let $a_j$ be the first agent (in order of preference by $a_i$) such that
        \begin{itemize}
            \item worst$_i(\Pi^{-1}(a_i))\succ_ia_j$, and
            \item $a_i\succ_j$worst$_j(\Pi^{-1}(a_j))$,
        \end{itemize}
        where worst$_i(\Pi^{-1}(a_i))$ is the worst predecessor of $a_i$ in $\Pi$ (and similarly for $a_j$), i.e., $\{a_i,a_j\}$ is $a_i$'s ``best blocking pair'' with regards to $\Pi$ in $I'$. Now
        \begin{enumerate}
            \item if there does not exist any such agent $a_j$, then clearly $\Pi'=\Pi\cup(a_i)$ is a GSP of $I'$. The intuition is the following: we increased $a_i$'s capacity, but no agent has an incentive (either due to free capacity or due to a worse partner in $\Pi$) to pair up with $a_i$, so $a_i$ remains with an additional unit of free capacity in the new GSP which is indicated by a fixed point $(a_i)$, in which case $\vert\mathcal{O}_{I'}^{\geq 3}\vert=\vert\mathcal{O}_{I}^{\geq 3}\vert$.

            \item if such an agent $a_j$ exists and $a_j$ is contained in a cycle $C=(a_j\;a_{r_1}\; a_{r_2} \dots a_{r_k})$ of odd length at least 3 in $\Pi$, then it is easy to show that $a_j$ is in no more than one such cycle (this was formally shown in \cite{glitzner2025unsolvabilitymanytomanynonbipartitestable}). Furthermore, in this case and by the fact that odd-length cycles of an instance are invariant due to Theorem 3.4, $\Pi'=(\Pi\setminus C)\cup(a_i\;a_j)(a_{r_1}\; a_{r_2})\dots(a_{r_{k-1}}\; a_{r_k})$ is a GSP of $I'$ with $\vert\mathcal{O}_{I'}^{\geq 3}\vert=\vert\mathcal{O}_{I}^{\geq 3}\vert-1$ as required. This follows by a similar argument as in the case above (1.). We give the following intuition: F1 holds for $\Pi'$ because it held for $\Pi$ by the assumption of $\Pi$ being a GSP for $I$, and the only cycles that are different in $\Pi'$ compared to $\Pi$ are transpositions, for which F1 is trivially satisfied. If F2 were violated in $\Pi'$, then there would be two agents that are not in a transposition in $\Pi'$ that strictly prefer each other to their worst predecessor in $\Pi'$. However, by inspection, these two agents must also violate F2 in $\Pi$ in $I$, a contradiction of $\Pi$ being a GSP. F3 holds by construction: it held for $\Pi$ in $I$, and we increased $a_i$'s capacity by 1; therefore, $a_i$ must be in one additional cycle in $\Pi'$. This is true: by the assumption that $a_i$ is not in a cycle of odd length at least 3, $a_i$ cannot be in $C$, and we add the transposition $(a_i\;a_j)$ which of course involves $a_i$, so $a_i$ is in exactly one more cycle in $\Pi'$ compared to $\Pi$ as required. All other agents remain in exactly as many cycles in $\Pi'$ as in $\Pi$ as required. By a similar argument, we can show that F4 holds for $\Pi'$ in $I'$: we know that F4 holds for $\Pi$ in $I$ by the assumption, and it was shown in \cite{glitzner2025unsolvabilitymanytomanynonbipartitestable} that two agents can only be adjacent in at most one cycle of a GSP. Thus, F4 cannot be violated by agents $a_{r_j}$ for any $1\leq j\leq k$. Agents $a_i,a_j$ also cannot violate F4 because by our assumption that worst$_i(\Pi^{-1}(a_i))\succ_ia_j$, it must be the case that $a_i,a_j$ were not previously adjacent in any cycle in $\Pi$. Therefore, adding the transposition $(a_i\;a_j)$ to $\Pi$ to gain $\Pi'$ does not violate F4.

            \item if such an agent $a_j$ exists but $a_j$ is not in any cycle of length longer than 2 in $\Pi$ (which would necessarily be an odd-length cycle due to the assumption that $\Pi$ is reduced), then, intuitively, the agent would reject their worst predecessor $a_w$ (which must be in a transposition $(a_j\;a_w)$ because $\Pi$ is assumed to be reduced) in $\Pi$ and enter the transposition $(a_i\;a_j)$ instead. This can, of course, have trickle down effects on $a_w$: when $a_w$ is a cycle $C=(a_w \; a_{r_1} \; a_{r_2}\dots a_{r_k})$ of odd length at least 3, then, by a similar argument as above, $\Pi'=(\Pi\setminus (C\cup(a_j\;a_w))\cup (a_i\;a_j)(a_w\;a_{r_1})(a_{r_2}\;a_{r_3})\dots(a_{r_{k}}\;a_{w})$ is a GSP of $I'$ with $\vert\mathcal{O}_{I'}^{\geq 3}\vert=\vert\mathcal{O}_{I}^{\geq 3}\vert-1$. When $a_w$ is not in such a cycle $C$, then one of three things can happen: $a_w$ remains in an additional fixed point $(a_w)$ in $\Pi'$ of $\Pi$, in which case $\vert\mathcal{O}_{I'}^{\geq 3}\vert=\vert\mathcal{O}_{I}^{\geq 3}\vert$, or $a_w$ is matched to the agent $a_q$ in their ``best blocking pair'' and the same scenarios occur for $a_q$ as previously for $a_j$ ($a_q$ can be either be contained in an odd-length cycle or only in transpositions -- either way we argued why the number of odd-length cycles cannot decrease by more than 1), in which case the number of odd-length cycles can either remain the same or decrease by 1, i.e., $\vert\mathcal{O}_{I'}^{\geq 3}\vert\geq\vert\mathcal{O}_{I}^{\geq 3}\vert-1$, or finally a new odd-length cycle can be formed, in which case $\vert\mathcal{O}_{I'}^{\geq 3}\vert\geq\vert\mathcal{O}_{I}^{\geq 3}\vert$.
        \end{enumerate} 
        
        \item Finally, suppose instead that $a_i$ is not in a cycle of odd length at least 3 in $\Pi$ and that $c_i'=c_i-1$. Again, the same scenarios as described above apply, i.e., either the number of odd cycles remains the same or is increased or decreased by 1, and thus a similar argument shows that indeed $\vert\mathcal{O}_{I'}^{\geq 3}\vert\geq\vert\mathcal{O}_{I}^{\geq 3}\vert-1$. 
    \end{enumerate}
    Therefore, we have now established the base case ($x=1$) of $P(x)$.
    
    Now consider the inductive case, that is, suppose that $P(x)$ holds and consider $P(x+1)$. Suppose, for the sake of contradiction, that there exists some $c'$ such that $\sum_{a_i\in A} \vert c_i'-c_i\vert=x+1$ and  $\vert\mathcal{O}_{I'}^{\geq 3}\vert< \vert\mathcal{O}_I^{\geq 3}\vert-(x+1)$. By the assumption that $P(x)$ is true, for any $c''$ such that $c''_i=c_i'$ for all $a_i\in A$ except some $a_j$ where either $c''_j=c_j'+1$ or $c''_j=c_j'-1$, it is true that $\vert\mathcal{O}_{I''}^{\geq 3}\vert\geq \vert\mathcal{O}_I^{\geq 3}\vert-x$. However, then it must be the case that one capacity change decreases the number of odd cycles of length at least 3 by at least 2, contradicting our earlier observations made above. Therefore, indeed, $\vert\mathcal{O}_{I'}^{\geq 3}\vert\geq \vert\mathcal{O}_I^{\geq 3}\vert-(x+1)$ holds as required, and we established the key property by induction.
\end{proof}

The following corollary observes that, in the worst case, the total amount of capacity modifications required for any of the three objectives concerned with the aggregate level is linear in the number of agents.

\begin{corollary}
    For any alternative capacity function $c'$ such that $c'$ is any of \emph{MAM$_{\pm}$}, \emph{MAM$_\text{+}$} or \emph{MAM$_{-}$}, the bound $\sum_{a_i\in A} \vert c_i'-c_i\vert=O(n)$ holds, where $n$ is the number of agents of the instance. This bound is tight for some families of instances.
\end{corollary}
\begin{proof}
    The upper bound $\sum_{a_i\in A} \vert c_i'-c_i\vert= O(n)$ follows directly from Theorem \ref{thm:mam}. For the lower bound, it suffices to note that \citet{glitzner2025empirics} observed the tight upper bound $\vert\mathcal{O}_I\vert\leq\left\lfloor\frac{n}{3}\right\rfloor+((n$ mod $3)$ mod $2)$ (where $\mathcal{O}_{I}$ includes both $\mathcal{O}_{I}^{\geq 3}$ and cycles of length 1) for {\sc sr}, so $\sum_{a_i\in A} \vert c_i'-c_i\vert= \Omega(n)$ follows from the fact that {\sc sr} is a special case of {\sc sf} and the tightness of the bound was demonstrated through the construction of an infinite family of instances.
\end{proof}

Conceptually, the sets of capacity functions that are optimal at the individual level intersect with those optimal at the aggregate level. While MIM$_+$ and MIM$_-$ functions are incompatible (unless the original instance is solvable), each has a non-empty intersection with MIM$_{\pm}$ functions. Similarly, MAM$_+$ and MAM$_-$ functions are incompatible, but each intersects with MAM$_{\pm}$ functions. Interestingly, MIM$+$ and MAM$+$ functions also share a non-empty intersection because both can be computed by \texttt{NearFeasible$_+$}, and similarly for the minus variants. Note, though, that it is not necessarily true that MIM$_\text{+}$ and MAM$_\text{+}$ functions (similarly for minus) are identical, and neither one is a subset of the other.

\section{Instability and Incentives}
\label{sec:instability}

While minimal capacity modifications ensure solvability, they may introduce substantial incentives for agents to deviate from the computed matching when evaluated against the original instance. In this section, we introduce tools to quantify and minimise these incentives, ensuring that computed solutions are both feasible and likely to be adhered to by self-interested agents.

\subsection{Motivating Observations}

In stable matching theory, the instability of a matching is usually quantified in terms of the number of blocking pairs and sometimes in terms of the number of distinct blocking agents involved in these pairs \cite{matchup,chen17}. Eriksson and Häggström \cite{Eriksson2008} argued that ``the proportion of blocking pairs is usually the best measure of instability'' and Manlove \cite{matchup} noted that the number of blocking pairs of a matching is considered proportional to the likelihood of a matching being undermined by a pair of agents \cite{matchup}. Specifically, a higher number of blocking pairs is considered to indicate a larger amount of instability due to more pairs of agents having both an incentive and an opportunity to cooperate and deviate from the matching, reorganising among themselves instead, and potentially causing unravelling among a larger set of agents. This could be due to agents having free capacity or due to them preferring each other to their respective worst partners (or a mixture of these between the two agents).

Consider, for a start, the number of blocking pairs introduced by decreasing an agent's capacity. To illustrate this, consider an {\sc sf} instance $I$ and an agent $a_i$ with the following preference list: 
$$a_i: a_{r_1}\;a_{r_2}\;a_{r_3}\;a_{r_4}\dots a_{r_{n-1}}$$
Now suppose that $c_i=2$ and let $\Pi$ be a GSP of $I$. Suppose, furthermore, that $a_i$ is in a transposition $(a_i\;a_{r_1})$ and in an odd-length cycle $C=(a_i \;a_{r_2}\;a_{s_1}\dots a_{s_k}\; a_{r_3})$ (for any agents $a_{s_j}$ such that $k$ is even) in $\Pi$. Recall from Theorem \ref{thm:gsp} that $C$ is invariant among all GSPs of $I$. In the previous section, we explained that the algorithm \texttt{NearFeasible$_\pm$} picks exactly one agent from $C$ arbitrarily and either increases or decreases their capacity by 1. \texttt{NearFeasible$_-$} also picks an arbitrary agent contained in $C$ and decreases their capacity (and, of course, \texttt{NearFeasible$_\text{+}$} increases their capacity). Now suppose that the algorithm picks $a_i$ as the agent to modify from $C$ and decreases their capacity, i.e., $c_i'=1$ in the alternative capacity function $c'$. Then, in the worst case, the near-feasible stable matching $M$ returned by \texttt{NearFeasible$_\pm$} (or \texttt{NearFeasible$_-$}) may admit a set of blocking pairs including all of $\{a_i,a_{r_3}\},\{a_i,a_{r_4}\},\dots,\{a_i,a_{r_{n-1}}\}$, i.e., $a_i$ contributes $n-3$ blocking pairs with respect to the original {\sc sf} instance $I$, where $c_i=2$. This is clearly undesirable as it indicates a very high incentive for $a_i$ (and also for the large set of agents blocking with $a_i$) to deviate from the computed matching $M$. The motivation for thinking about the agent incentives with respect to $I$ rather than with respect to the new instance corresponding to the new capacity function $c'$ is practical in nature: we assume that agents report their preferences and capacities truthfully in $I$ and then compute a solution in $I'$ which might involve changed capacities, but we did not \emph{convince} agents to change their capacities. Instead, we made \emph{choices} about whose capacities to change. To compute a solution that agents will (ideally even \emph{want} to) comply with, we shall aim to decrease the agents' incentives to reorganise themselves after the release of the computed matching.

The statement below generalises our observation about the potential instability of near-feasible stable matchings asymptotically.

\begin{prop}
    Let $M$ be a near-feasible stable matching with respect to some \emph{MAM$_{\pm}$} or \emph{MAM$_{-}$} alternative capacity function. Then, with respect to the original instance without modified capacities, the bound $\vert bp(M)\vert=O(n^2)$ applies, where $n$ is the number of agents of the instance. This bound is tight for some families of instances.
\end{prop}
\begin{proof}
    Clearly $bp(M)\leq \frac{n(n-1)}{2}$ by the number of possible unordered pairs of agents in the instance. Now to see that $bp(M)=\Omega(n^2)$ in the worst case, consider an infinite family of instances with $n\in 3\mathbb{N}$ agents, $c_i=1$ for all $a_i\in A$, and with preferences constructed as follows: partition the set of agents into sets of size 3 and construct the preferences of these agents within each set such that they form cycles of length 3 in any GSP $\Pi$ of the instance (i.e., directed preference cycles). Then append all remaining agents to the end of their preference lists in arbitrary order. Clearly $\vert\mathcal{O}_I^{\geq 3}\vert=\tfrac{n}{3}$, so in a MAM$_{-}$ capacity function $c'$ computed using {\tt NearFeasible$_-$} decreases the capacity of $\tfrac{n}{3}$ agents and leaves them unmatched in the resulting near-feasible stable matching $M$. However, by complete preferences, all these unmatched agents form blocking pairs with each other with respect to the original instance, i.e., $\vert bp(M)\vert \geq \tfrac{n/3(n/3-1)}{2}=\tfrac{1}{18}n^2-\tfrac{1}{6}n$ and hence the lower bound for MAM$_{-}$ follows. Notice that we showed that any MAM$_{-}$ capacity function is a MAM$_{\pm}$ capacity function in Theorem \ref{thm:mam}, so the result follows.
\end{proof}

Now, on the other hand, if $a_i$'s capacity were increased rather than decreased, then $a_i$ could be matched to all of $a_{r_1}, a_{r_2}$ and $a_{r_3}$ in the computed matching $M$, i.e., $a_i$'s true reported capacity would be exceeded by 1. Although the solution is technically not a valid matching in the original instance, from the perspective of blocking pairs, we detect no issues: $a_i$ does not contribute any blocking pairs to $bp(M)$ with respect to the original instance. Still, $a_i$ clearly has an \emph{incentive} not to engage with $a_{r_3}$ (e.g., due to limited time or other limited resources) because they reported that they can, or at least prefer to, have at most two matches, but they are now matched to three agents and $a_{r_3}$ is the worst among them. Arguably, agents $a_i,a_{r_3}$ may be blocking \emph{asymmetrically}, i.e., it could be the case that $a_i$ would like to drop $a_{r_3}$ from its set of matches, but $a_{r_3}$ does not want to drop $a_i$. This asymmetry is crucial in practice: an agent may technically be matched beyond their reported capacity, but the agent may resist cooperating with all assigned matches, leading to instability. This raises the following key question, which we will answer in the remainder of this section: what is a suitable trade-off between capacity increases and decreases, and what impact does it have on the instability of the resulting matching with regard to the original problem instance?

\subsection{Blocking Entries and Optimality}

We begin by extending the classical notion of a blocking pair to this asymmetrical phenomenon as a quantitative measure of instability and refer to it as a \emph{blocking entry}. A formal definition follows below.

\begin{definition}[Blocking Entry]
    Let $I=(A,\succ, c)$ be an {\sc sf} instance, let $c'$ be an alternative capacity function for $I$, and let $M$ be a matching in $I'=(A,\succ,c')$. Then, for two distinct agents $a_i, a_j\in A$, $a_j$ is a \emph{blocking entry} of $a_i$, denoted by an ordered tuple $(a_i,a_j)$, if 
    \begin{enumerate}[label=(\roman*),leftmargin=2em]
        \item $\{a_i,a_j\}\in M$ and $\vert \{a_r\in M(a_i) \;\vert\; a_r\succ_i a_j\}\vert \geq c_i$, or 
        \item $\{a_i,a_j\}\notin M$, $\vert\{a_r\in M(a_i)\;\vert\;a_r\succ_i a_j\}\vert < c_i$ and $\vert\{a_r\in M(a_j)\;\vert\;a_r\succ_j a_i\}\vert < c_j$.
    \end{enumerate}
    Let $be_i(M)$ denote the set of blocking entries of $a_i$ admitted by $M$ and let $be(M)=\bigcup_{a_i\in A}be_i(M)$.
\end{definition}

Intuitively, a blocking entry represents a single agent's incentive to either drop or add a match, extending the classical notion of blocking pairs to capture asymmetries due to capacity violations. Notice that while the blocking entries $(a_i,a_j)$ and $(a_j,a_i)$ need not be equivalent, if $(a_i,a_j)\in be(M)$ due to condition (ii), then also $(a_j,a_i)\in be(M)$. At this point, it is important to understand the differences between blocking pairs and blocking entries. Part (ii) of the definition mirrors precisely the condition of a blocking pair, and part (i) extends the definition to agents blocking because someone's capacity is exceeded. Conceptually, blocking entries quantify the amount of instability of a matching by counting the number of misaligned agent incentives between reported preferences and capacities and the computed matching. Fundamentally, blocking entries treat positive incentives (to engage with someone that they are not matched to) the same as negative incentives (to disengage with someone that they are matched to). In particular, if a matching is stable, then it does not admit any blocking pairs or any blocking entries (and vice versa). Thus, an instance is solvable if and only if there exists a matching with no blocking entries. 

We will now introduce the following objective that formally captures our goal of minimising agent incentives to deviate at the individual level.

\begin{definition}
    Let $I=(A,\succ,c)$ be an {\sc sf} instance, let $C$ be a set of alternative capacity functions, let $M$ be a matching in $I'=(A,\succ,c')$ for some $c'\in C$, and let $k=\max_{a_i\in A}\vert be_i(M)\vert$. Then $M$ is a \emph{minimal individual deviation incentive} (\emph{MIDI}) matching with respect to $C$ (which we can also denote by \emph{MIDI}$_C$) if for all $c''\in C$ and for all matchings $M'$ of $I''=(A,\succ,c'')$, it is the case that $k\leq \max_{a_i\in A}\vert be_i(M')\vert$.
\end{definition}

MIDI matchings minimise the maximum incentive of any individual agent to deviate among all feasible matchings, where feasibility is controlled through a collection of capacity functions $C$. This is critical in multi-agent settings where a single strongly dissatisfied agent can trigger unravelling. 

Similar to the previous section, we now introduce an analogous optimality notion aimed at the aggregate level, which we will refer to as MADI. MADI matchings minimise the total incentive to deviate, aiming to stabilise the system as a whole.

\begin{definition}
    Let $I=(A,\succ,c)$ be an {\sc sf} instance, let $C$ be a set of alternative capacity functions, let $M$ be a matching in $I'=(A,\succ,c')$ for some $c'\in C$, and let $k=\vert be(M)\vert$. Then $M$ is a \emph{minimal aggregate deviation incentive} (\emph{MADI}) matching with respect to $C$ (also denoted by \emph{MADI}$_C$) if, for all $c''\in C$ and for all matchings $M'$ of $I''=(A,\succ,c'')$, it is the case that $k\leq \vert be(M')\vert$.
\end{definition}

\subsection{Minimising Blocking Entries without Capacity Violations}
\label{sec:midimadic}

Consider the role of the set $C$ of alternative capacity functions for the definitions of MIDI and MADI. We start by analysing two extremes: respecting original capacities strictly, and allowing arbitrary capacity modifications. For now, notice that if $C=\{c\}$, i.e., the set consists only of the original capacity function provided as part of the input and no alternatives, then no MIDI$_{\{c\}}$ or MADI$_{\{c\}}$ matching may violate any of the original capacities. Hence, the notion of blocking entries simply breaks down to that of blocking pairs and, therefore, a MIDI$_{\{c\}}$ matching is a matching that minimises the maximum number of blocking pairs that any agent is involved in, and a MADI$_{\{c\}}$ matching is an \emph{almost-stable} matching (which has been studied in the {\sc Stable Roommates} setting, for example, in \cite{abraham06,chen2019computational,chen17}, as previously outlined in Section \ref{sec:relatedwork}). To establish the complexity of computing such optimal matchings, we first note the following property, which follows directly from the definitions of blocking pairs and blocking entries.

\begin{prop}
    \label{prop:madicclaim2}
     Let $I=(A,\succ, c)$ be an {\sc sf} instance and let $M$ be a matching in a near-feasible instance $I'=(A,\succ,c')$. If $\{a_i,a_j\}\in bp(M)$ then $\{(a_i,a_j),(a_j,a_i)\}\subseteq be(M)$. Thus $2\vert bp(M)\vert \leq \vert be(M)\vert$. If $c=c'$, then $\{a_i,a_j\}\in bp(M)$ if and only if $(a_i,a_j) \in be(M)$ and $(a_j,a_i) \in be(M)$. Thus, if $c=c'$, then $2\vert bp(M)\vert = \vert be(M)\vert$ and $2\vert bp_i(M)\vert = \vert be_i(M)\vert$.
\end{prop}
    
Now, we can immediately establish the intractability of computing a MIDI$_{\{c\}}$ matching.

\begin{theorem}
\label{thm:midichard}
    Let $I=(A,\succ,c)$ be an {\sc sf} instance. The problem of finding a \emph{MIDI}$_{\{c\}}$ matching in $I$ is {\sf NP-hard} with respect to the optimum value.
\end{theorem}
\begin{proof}
    It suffices to show the following claim: a matching $M$ of $I$ is \emph{MIDI}$_{\{c\}}$ if and only if, for all matchings $M'$ of $I$, it is true that $\max_{a_i\in A}\vert bp_i(M)\vert\leq \max_{a_i\in A}\vert bp_i(M')\vert$.

    Clearly $M$ is MIDI$_{\{c\}}$ if, for all matchings $M'$ of $I$, it is the case that $\max_{a_i\in A}\vert bp_i(M)\vert\leq \max_{a_i\in A}\vert bp_i(M')\vert$. Proposition \ref{prop:madicclaim2} established that if $M$ is a matching in $I$, then $2\vert bp_i(M)\vert=\vert be_i(M)\vert$. Hence, minimising the maximum number of blocking entries among all agents simultaneously minimises the maximum number of blocking pairs, and vice versa. More specifically, $M$ is MIDI$_{\{c\}}$ if and only if $2\max_{a_i\in A}\vert bp_i(M)\vert\leq 2\max_{a_i\in A}\vert bp_i(M')\vert$, and clearly $2\max_{a_i\in A}\vert bp_(M)\vert\leq 2\max_{a_i\in A}\vert bp_i(M')\vert$ if and only if $\max_{a_i\in A}\vert bp_i(M)\vert\leq \max_{a_i\in A}\vert bp_i(M')\vert$. The claim follows.

    This immediately implies that computing a MIDI$_{\{c\}}$ matching is {\sf NP-hard}, taking advantage of the {\sf NP-hardness} result for the problem of finding a matching that minimises the maximum number of blocking pairs that any agent is in, which was established by \citet{glitznermanloveminimax} even in the {\sc sr} setting.
\end{proof}

We note that our simple reduction based on the result in \cite{glitznermanloveminimax} implies that even deciding whether a MIDI$_{\{c\}}$ exists where no agent is in more than one blocking entry is {\sf NP-complete}, so the problem of finding a MIDI$_{\{c\}}$ matching cannot be in {\sf XP} with respect to the optimal parameter value, unless {\sf P}={\sf NP}. This will contrast with our result for MADI$_{\{c\}}$ matchings in Section \ref{sec:xp}.

We will now show formally that MADI$_{\{c\}}$ matchings are equivalent to almost-stable matchings and thus also intractable to compute.

\begin{theorem}
\label{thm:madichard}
    Let $I=(A,\succ,c)$ be an {\sc sf} instance. The problem of finding a \emph{MADI}$_{\{c\}}$ matching in $I$ is {\sf NP-hard}.
\end{theorem}
\begin{proof}
    It suffices to show the following claim: a matching $M$ of $I$ is \emph{MADI}$_{\{c\}}$ if and only if it is almost-stable, i.e., if for all matchings $M'$ of $I$, it is true that $\vert bp(M)\vert\leq \vert bp(M')\vert$.
    
    Clearly $M$ is MADI$_{\{c\}}$ if, for all $M'$ of $I$, it is the case that $\vert be(M)\vert\leq \vert be(M')\vert$. Proposition \ref{prop:madicclaim2} established that if $M$ is a matching in $I$, then $2\vert bp(M)\vert=\vert be(M)\vert$. Hence, minimising the number of blocking entries simultaneously minimises the number of blocking pairs and vice versa. More specifically, $M$ is MADI$_{\{c\}}$ if and only if $2\vert bp(M)\vert\leq 2\vert bp(M')\vert$, and clearly $2\vert bp(M)\vert\leq 2\vert bp(M')\vert$ if and only if $\vert bp(M)\vert\leq \vert bp(M')\vert$. The claim follows.

    This immediately implies that computing a MADI$_{\{c\}}$ matching is {\sf NP-hard}, taking advantage of the {\sf NP-hardness} result for the problem of finding an almost-stable matching in {\sc sr}, which is a special case of {\sc sf}, established by \citet{abraham06}.
\end{proof}

We will return to these intractability results in Section \ref{sec:exact}.

\subsection{Minimising Blocking Entries in General}
\label{sec:midimadiall}

We now consider the other extreme, where the set $C$ contains all possible alternative capacity functions and we aim to minimise the individual or aggregate incentives to deviate. This is a natural generalisation to study, as it presumes that we want to optimise with respect to the set of all possible pairings of agents, rather than restricting ourselves to matchings that satisfy original agent capacities. We will refer to such optimal matchings for this choice of $C$ as MIDI$_{\infty}$ and MADI$_{\infty}$.\footnote{Note that $C$ only contains functions $c$ that respect $0\leq c_i\leq n-1$ for all agents $a_i$, so $\infty$ indicates a generally large but finite set $C$.}

We first consider the properties and computation of MIDI$_{\infty}$ matchings and show that a well-balanced matching, in the sense that no agent needs to be involved in more than one blocking entry, always exists and can be computed efficiently. 

\begin{prop}
\label{prop:midigeneral}
    Let $I=(A,\succ,c)$ be an {\sc sf} instance. Then any \emph{MIDI}$_{\infty}$ matching $M$ satisfies $\max_{a_i\in A}\vert be_i(M)\vert=\min(\vert \mathcal{O}^{\geq 3}_I\vert,1)$ and one can be computed efficiently using {\tt NearFeasible$_\text{+}$}.
\end{prop}
\begin{proof}
    Consider the matching $M$ computed using {\tt NearFeasible$_+$}. $M$ does not admit any blocking pairs, as $M$ is stable with respect to the near-feasible instance and any modified capacities are increased. Hence, $M$ admits no blocking entries of type (ii). Furthermore, no agent's capacity is increased by more than 1 during the computation of $M$; therefore no agent has more than one blocking entry of type (i). Thus $\max_{a_i\in A}\vert be_i(M)\vert\leq 1$. Optimality of {\tt NearFeasible$_+$} follows easily: $I$ is solvable if and only if $\vert \mathcal{O}^{\geq 3}_I\vert=0$. Thus, if $I$ is solvable then {\tt NearFeasible$_+$} returns a stable matching in $I$ so $be(M)=\varnothing$ and hence $\max_{a_i\in A}\vert be_i(M)\vert=0=\vert \mathcal{O}^{\geq 3}_I\vert$. If $I$ is unsolvable instead, then any MIDI$_{\infty}$ matching $M'$ will have $\max_{a_i\in A}\vert be_i(M')\vert\geq 1\geq \max_{a_i\in A}\vert be_i(M)\vert$ as required.
\end{proof}

This positive tractability result is promising, given the intractability results for the computation of MIDI$_{\{c\}}$ and MADI$_{\{c\}}$ and various other almost-stability problems that we highlighted previously. Even more surprisingly, though, we can also give a positive tractability result and tight characterisation of the optimal value for the computation of MADI$_{\infty}$ matchings, as the following statement shows.

\begin{theorem}
\label{thm:madigeneral}
    Let $I=(A,\succ,c)$ be an {\sc sf} instance. Then any \emph{MADI}$_{\infty}$ matching $M$ has $\vert be(M)\vert=\vert \mathcal{O}^{\geq 3}_I\vert$ and one can be computed efficiently using {\tt NearFeasible$_\text{+}$}.
\end{theorem}
\begin{proof}
    We have already shown in Proposition \ref{prop:midigeneral} that {\tt NearFeasible$_+$} computes a matching in which every one of the $\vert \mathcal{O}^{\geq 3}_I\vert$ agents whose capacity is modified has exactly one blocking entry and no other agent has any blocking entry. Thus the upper bound $\vert be(M)\vert\leq \vert \mathcal{O}^{\geq 3}_I\vert$ is immediate. It remains to show that also $\vert be(M)\vert\geq \vert \mathcal{O}^{\geq 3}_I\vert$. We will first prove the following claim: for every matching $M_g$ of $I$ that, however, might not respect all capacities $c$, there exists a solvable {\sc sf} instance $I'=(A,\succ,c')$ such that $\sum_{a_i\in A}\vert c_i'-c_i\vert\leq \vert be(M_g)\vert$.

    We will establish this claim by constructing such an instance $I'$ and a stable matching $M_s$ of $I'$ explicitly. To start, let $c'=c$ and let $M_s=M_g$. Now consider each blocking entry $e=(a_i,a_j)\in be^I(M_g)$ (i.e., each blocking entry that $M_g$ admits with respect to $I$) iteratively, starting with blocking entries of type (ii), i.e., the blocking pairs. For every $(a_i,a_j)\in be^I(M_g)$ of type (ii), increase $c_i'$ by 1 and, if not yet present, add $\{a_i,a_j\}$ to $M_s$. Now we are only left with blocking entries of type (i), i.e., those that correspond to matches in $M_g$ (and hence in $M_s$ for now) but violate the capacity constraint of the respective agent. For every such $(a_i,a_j)\in be^I(M_g)$, if also $(a_j,a_i)\in be^I(M_g)$ then remove $\{a_i,a_j\}$ from $M_s$ if not yet removed and leave $c'$ unchanged. If instead $(a_i,a_j)\in be^I(M_g)$ but $(a_j,a_i)\notin be^I(M_g)$, then remove $\{a_i,a_j\}$ from $M_s$ and decrease $c_j'$ by 1 (the intuition here is that when the match with $a_j$ violates $a_i$'s capacity but not vice versa, then by simply removing the match and decreasing $a_j$'s capacity we do not impact the relative blocking potential of $a_j$).
    
    We claim that $be^{I'}(M_s)=\varnothing$ (i.e., $M_s$ admits no blocking entries with respect to $I'$). Clearly, by discarding all worst matches that exceed an agent's capacity, increasing capacities with every new match, and decreasing capacities only with the removal of a match that did not violate an agent's capacity, $M_s$ respects all capacities $c'$ by construction. Hence, $M_s$ does not admit any type (i) blocking entries with respect to $I'$. Now suppose, for the sake of contradiction, that there exists a type (ii) blocking entry $(a_p,a_q)\in be(M_s)$ with respect to $I'$. Then, by definition, $\{a_p,a_q\}\notin M_s$, $\vert\{a_r\in M_s(a_p)\;\vert\;a_r\succ_p a_q\}\vert < c_p'$ and $\vert\{a_r\in M_s(a_q)\;\vert\;a_r\succ_q a_p\}\vert < c_q'$, i.e., $\{a_p,a_q\}$ is a blocking pair of $M_s$, and indeed we will show the contradiction that $\{a_p,a_q\}$ is a blocking pair of $M_g$. 
    
    Let us start by establishing that an agent has free capacity in $M_g$ (with respect to $c$) if and only if they have free capacity in $M_s$ (with respect to $c'$). To see this, consider the following. Suppose first that $\vert M_g(a_i)\vert<c_i$. We know that $\vert M_s(a_i)\vert = \vert M_g(a_i)\vert -t_1+ t_2$, where $t_1$ indicates the number of removed matches due to type (i) blocking entries and $t_2$ indicates the number of new matches due to type (ii) blocking entries. However, if $\vert M_g(a_i)\vert<c_i$ then $a_i$ cannot have any type (i) blocking entries, so any removal of matches due to some other agent's type (i) blocking entry comes with a capacity decrease for $a_i$ of 1. Hence, $c_i'= c_i -t_1 + t_2$ by construction. By rearranging, we have that $\vert M_s(a_i)\vert +t_1-t_2 = \vert M_g(a_i)\vert <c_i=c_i'+t_1-t_2$, so $\vert M_s(a_i)\vert < c_i'$ as required. 
    
    To show the converse direction, suppose that $\vert M_g(a_i)\vert\geq c_i$. Again, clearly $\vert M_s(a_i)\vert = \vert M_g(a_i)\vert -t_1+ t_2$. To be more specific, we can separate $t_1=t_1^a+t_1^b$ into two parts, where $t_1^a$ is the number of removed matches due to type (i) blocking entries of the form $(a_i,a_j)$ and $t_1^b$ is the number of removed matches due to type (i) blocking entries of agents of the form $(a_j,a_i)$ (such that $(a_i,a_j)$ is not a blocking entry). Clearly $\vert M_g(a_i)\vert= c_i+t_1^a$, so $\vert M_s(a_i)\vert = \vert M_g(a_i)\vert -t_1+ t_2=c_i+t_1^a-t_1+t_2=c_i+t_1^a-t_1^a-t_1^b+ t_2=c_i-t_1^b+ t_2$. Furthermore, by construction, $c_i'=c_i-t_1^b+t_2$. Rearranging, we get $c_i=c_i'+t_1^b-t_2$, and therefore $\vert M_s(a_i)\vert =c_i'+t_1^b-t_2-t_1^b+ t_2=c_i'$ as required.

    Now, let us return to the blocking pair $\{a_p,a_q\}$ admitted by $M_s$. We note that $\{a_p,a_q\}\notin M_g$. To see this, suppose that $\{a_p,a_q\}\in M_g$, then, by construction of $M_s$, because $\{a_p,a_q\}\notin M_s$ (by properties of blocking pairs), it must be the case that $\{a_p,a_q\}$ was removed because it was a type (i) blocking entry either of $a_p$ or of $a_q$, or both. However, then either $\vert\{a_r\in M_g(a_p)\;\vert\;a_r\succ_p a_q\}\vert \geq c_p$ or $\vert\{a_r\in M_g(a_q)\;\vert\;a_r\succ_q a_p\}\vert \geq c_q$ (or both). Without loss of generality, suppose that $\vert\{a_r\in M_g(a_p)\;\vert\;a_r\succ_p a_q\}\vert \geq c_p$. 

    Note that, by construction of $M_s$, for every agent $a_i$ that admits a type (i) blocking entry in $M_g$, it holds that $w_i(M_s)\succeq_i a_j$, where $w_i(M_s)$ is the worst agent matched to $a_i$ in $M_s$ (according to $\succ_i$) and $a_j$ is the agent such that $\vert \{a_r\in M_g(a_i)\;\vert\;a_r\succeq_i a_j\}\vert=c_i$. This is because if $a_i$ admits a type (i) blocking entry then $\vert M_g(a_i)\vert>c_i$ (this also justifies the existence of $a_j$); therefore any match in $M_g(a_i)$ worse than $a_j$ (according to $\succ_i$) is discarded and any new match in $M_s(a_i)\setminus M_g(a_i)$ must stem from a blocking pair $\{a_i,a_k\}$ with respect to $M_g$, in which case it is true that $a_k\succ_i w_i(M_s)$ as required.

    Thus, by $\vert\{a_r\in M_g(a_p)\;\vert\;a_r\succ_p a_q\}\vert \geq c_p$, there exists such an agent $a_j$ such that $\vert \{a_r\in M_g(a_p)\;\vert\;a_r\succeq_p a_j\}\vert=c_p$ and $a_j\succ_i a_p$. Therefore, by the fact we established above, $w_p(M_s)\succeq_p a_q$, contradicting that $\{a_p,a_q\}\in bp(M_s)$. Hence, $\{a_p,a_q\}\notin M_g$ as required.
    
    Now, 
    \begin{itemize}
        \item if both $a_p$ and $a_q$ have free capacity in $M_s$ (with respect to $c'$), then they also have free capacity in $M_g$ (with respect to $c$), so $\vert\{a_r\in M_g(a_p)\;\vert\;a_r\succ_p a_q\}\vert \leq \vert M_g(a_p)\vert < c_p$ and $\vert\{a_r\in M_g(a_q)\;\vert\;a_r\succ_q a_p\}\vert \leq \vert M_g(a_q)\vert < c_q$, in which case $\{a_p,a_q\}$ must be a blocking pair of $M_g$, contradicting the fact that all such pairs are matches in $M_s$ by construction.
        
        \item if neither $a_p$ nor $a_q$ has free capacity in $M_s$ (with respect to $c'$), then the agents must prefer each other to their worst partners in $M_s$. Let $w_p(M_s)$ and $w_q(M_s)$ denote these worst partners, respectively. Then, by construction of $M_s$, either $\{a_p,w_p(M_s)\}\in M_g$ or $\{a_p,w_p(M_s)\}\in bp(M_g)$. In the former case, it must be true that $\vert\{a_r\in M_g(a_p)\;\vert\;a_r\succ_p w_p(M_s)\}\vert < c_p$ (otherwise $(a_p,w_p(M_g))$ would be a type (i) blocking entry and so, by construction, $\{a_p,w_p(M_g)\}\notin M_s$, a contradiction) and, in the latter case, it must be true that $\vert\{a_r\in M_g(a_p)\;\vert\;a_r\succ_p w_p(M_s)\}\vert < c_p$. Hence, by $a_q\succ_p w_p(M_s)$, in either case it is true that $\vert\{a_r\in M_g(a_p)\;\vert\;a_r\succ_p a_q\}\vert < c_p$. By the same reasoning, it must be the case that $\vert\{a_r\in M_g(a_q)\;\vert\;a_r\succ_q a_p\}\vert < c_q$. However, then $\{a_p,a_q\}$ is a blocking pair of $M_g$, again leading to the same contradiction that all such pairs are matches in $M_s$ by construction.
        
        \item if (without loss of generality) $a_p$ has free capacity in $M_s$ (with respect to $c'$) but $a_q$ does not, then $a_p$ also has free capacity in $M_g$ (with respect to $c$). Hence, $\vert\{a_r\in M_g(a_p)\;\vert\;a_r\succ_p a_q\}\vert \leq \vert M_g(a_p)\vert < c_p$. Furthermore, we already argued in the point above why, in this situation, it must be the case that $\vert\{a_r\in M_g(a_q)\;\vert\;a_r\succ_q a_p\}\vert < c_q$. Then, again, $\{a_p,a_q\}$ must be a blocking pair of $M_g$, leading to the same contradiction that all such pairs are matches in $M_s$ by construction.
    \end{itemize}    
    Therefore, $be^{I'}(M_s)=\varnothing$ and clearly $\sum_{a_i\in A}\vert c_i'-c_i\vert\leq \vert be^I(M_g)\vert$ as required.

    Now to return to the lower bound $\vert be(M)\vert\geq \vert \mathcal{O}^{\geq 3}_I\vert$ that we set out to prove, suppose for the sake of contradiction that $\vert be(M)\vert< \vert \mathcal{O}^{\geq 3}_I\vert$. Then by our claim above, there exists a solvable {\sc sf} instance $I'=(A,\succ,c')$ such that $\sum_{a_i\in A}\vert c_i'-c_i\vert\leq \vert be(M)\vert<\vert\mathcal{O}^{\geq 3}_I\vert$, which contradicts the fact that any capacity modification that results in a solvable instance requires at least $\vert\mathcal{O}_I^{\geq 3}\vert$ capacity changes, as we previously established in Theorem \ref{thm:mam}. This completes the proof.
\end{proof}

Our positive results above show that optimality on the individual and aggregate levels with regard to minimum deviation incentive is compatible and -- surprisingly -- a matching that satisfies both criteria simultaneously can be computed in polynomial time using the rather simple algorithm {\tt NearFeasible$_\text{+}$}.

\subsection{Compatibility of Minimal Modifications and Minimal Deviation Incentives}
\label{sec:compatibility}

In Section \ref{sec:midimadic}, we considered the computation of MIDI and MADI matchings when all agent capacities need to be respected. In Section \ref{sec:midimadiall}, we considered the computation of such a matching when none of the capacities need to be respected. However, while the former was too restrictive to get polynomial-time algorithms, the latter turned out to be accommodating enough to allow for efficient optimal algorithms. On the other hand, allowing arbitrary capacity violations in favor of minimum instability can have a major impact on participating individuals whose capacity is strongly modified. Therefore, we now return to minimal capacity modifications from Section \ref{sec:midimadic}. How do they perform with respect to our new measures of instability and deviation incentives? Restricting the set of alternative capacities to minimal modifications, we obtain the following result for MIDI matchings.

\begin{prop}
\label{prop:nearfeasiblemidi}
    {\tt NearFeasible$_\text{+}$} computes a matching that has minimal individual deviation incentive while simultaneously requiring minimal individual and aggregate capacity modifications, regardless of whether capacities can only be increased or also decreased. Formally, {\tt NearFeasible$_\text{+}$} computes a matching that is \emph{MIDI}$_{X}$ for $X\in\{\text{\emph{MIM}}_{\pm},\text{\emph{MIM}}_\text{+},\text{\emph{MAM}}_{\pm},\text{\emph{MAM}}_\text{+}\}$. However, {\tt NearFeasible$_-$} does not necessarily return a \emph{MIDI$_{\text{MIM}_{-}}$} or a \emph{MIDI$_{\text{MAM}_{-}}$} matching.
\end{prop}
\begin{proof}
    In Proposition \ref{prop:midigeneral} we showed that {\tt NearFeasible$_+$} returns a MIDI$_\infty$ matching and we know by Proposition 3.5 that {\tt NearFeasible$_+$} returns a matching that respects some MIM$_{+}$ capacity function. Again by Proposition 3.5, any MIM$_{+}$ capacity function is a MIM$_{\pm}$ capacity function. We also know by Theorem \ref{thm:mam} that {\tt NearFeasible$_+$} returns a matching that respects some MAM$_{+}$ capacity function. Again by Theorem \ref{thm:mam}, any MAM$_{+}$ capacity function is a MAM$_{\pm}$ capacity function. Hence, optimality of {\tt NearFeasible$_+$} for MIM$_{\pm}$, MIM$_{+}$, MAM$_{\pm}$ and MAM$_{+}$ follows immediately.
    
    To show that {\tt NearFeasible$_-$} does not necessarily return a MIDI$_{\text{MIM}_{-}}$ or a MIDI$_{\text{MAM}_{-}}$ matching, consider an instance containing agents $a_i,a_j$ with preferences such that they are contained in the same cycle of odd length at least 3, but one has their predecessor of this odd cycle as their second choice, while the other has their predecessor of this odd cycle as their last choice. Then it is entirely possible that {\tt NearFeasible$_-$} decreases either $c_i$ or $c_j$ by 1, but while in the former case this might cause linearly many blocking entries for $a_i$ in the number of agents, the latter case causes just one blocking entry for $a_j$. Recall that we showed in Proposition 3.5 that {\tt NearFeasible$_-$} always computes a MIM$_{-}$ capacity function, but now we have shown that, contrary to the procedure of {\tt NearFeasible$_-$}, the choice of which agents' capacities to decrease cannot always be made arbitrarily when aiming to compute a MIDI matching with respect to the MIM$_{-}$ or the MAM$_{-}$ capacity functions.
\end{proof}

The result above shows that among all alternative capacity functions that change the agent capacities by a minimal amount (either only upwards or in either direction), and among all matchings for instances with these alternative capacity functions, a matching that induces minimal individual incentive to deviate with respect to the original instance can be found efficiently using {\tt NearFeasible$_\text{+}$}. However, it also notes that when restricting ourselves to downwards capacity modifications only, {\tt NearFeasible$_-$} does not necessarily find a MIDI matching and does not even give any desirable approximation guarantees -- something we could already expect from the intractability results in Theorems \ref{thm:midichard}-\ref{thm:madichard}. We note that MIDI$_{\text{MIM}_{-}}$ and MIDI$_{\text{MAM}_{-}}$ are rather unnatural notions, as the motivation for them was to never exceed reported agent capacities -- but this is already captured by MIDI$_{\{c\}}$.

Now, let us state similar observations for minimal capacity modifications and minimum deviation incentive at the aggregate level.

\begin{prop}
\label{prop:nearfeasiblemadi}
    {\tt NearFeasible$_\text{+}$} computes a matching that has minimal aggregate deviation incentive while simultaneously requiring minimal individual and aggregate capacity modifications, regardless of whether capacities can only be increased or also decreased. Formally, {\tt NearFeasible$_\text{+}$} computes a matching that is \emph{MADI}$_{X}$ for $X\in\{\text{\emph{MIM}}_{\pm},\text{\emph{MIM}}_\text{+},\text{\emph{MAM}}_{\pm},\text{\emph{MAM}}_\text{+}\}$. However, {\tt NearFeasible$_-$} does not necessarily return a \emph{MADI$_{\text{MIM}_{-}}$} or a \emph{MADI$_{\text{MAM}_{-}}$} matching.
\end{prop}
\begin{proof}
    The proof is identical to that of Proposition \ref{prop:nearfeasiblemidi}, but invoking Theorem \ref{thm:madigeneral} in the places where Proposition \ref{prop:midigeneral} was invoked previously.
\end{proof}

This result again implies compatibility between minimal capacity modifications and minimal deviation incentives in all settings except when only decreasing capacities. Again, though, MADI$_{\text{MIM}_{-}}$ and MADI$_{\text{MAM}_{-}}$ are rather unnatural and  MADI$_{\{c\}}$ would be a more natural aim in this context.

\section{Capacity-Respecting Exponential-Time Algorithms}
\label{sec:exact}

Although most problems of interest turned out to be efficiently solvable in the previous section, the two problems that remained intractable due to direct connections to almost-stable matching problems are those aiming to find a MIDI$_{\{c\}}$ and a MADI$_{\{c\}}$ matching. We will consider two approaches of solving these problems to optimality despite the computational boundaries we identified previously.

\subsection{An {\sf XP} Algorithm for MADI$_{\{c\}}$}
\label{sec:xp}

While we noted in Section \ref{sec:midimadic} why the problem of finding a MIDI$_{\{c\}}$ matching is not in {\sf XP} with respect tot he optimum value (i.e., the minimum maximum number of blocking entries per agent) unless {\sf P}$=${\sf NP}, we will now show that the problem of finding a MADI$_{\{c\}}$ matching, instead, is in {\sf XP} with respect to the optimum value (i.e., the minimum number of blocking entries across all agents). Algorithm \ref{alg:minbp} is a modified brute-force style algorithm and proceeds as follows: it enumerates all possible subsets (of increasing size) of pairs of agents that could be blocking, then, for each candidate subset, makes these potentially blocking pairs of agents mutually unacceptable, and checks whether the resulting instance admits a stable matching using the algorithm by \citet{Irving2007}. The first stable matching encountered is returned as the solution. Note that from hereon, we denote the set of all matchings of an instance by $\mathcal M$.

\begin{algorithm}[!htb]
\renewcommand{\algorithmicrequire}{\textbf{Input:}}
\renewcommand{\algorithmicensure}{\textbf{Output:}}

\begin{algorithmic}[1]
\Require{$I=(A,\succ,c)$ : an {\sc sf} instance}
\Ensure{$M$ : a matching}

\For{$k\in\{0,1,\dots,\vert A\vert(\vert A\vert-1)\}$ in increasing order}
    \For{set $B$ of $k$ pairs of distinct agents in $A$}
        \State make copy $I'$ of $I$
        \State make all pairs in $B$ mutually unacceptable in $I'$
        \State $M\gets$ {\sf IrvingScott}$(I')$
        \If{$M$ is not None}
            \State\Return{$M$}
        \EndIf
    \EndFor
\EndFor
\end{algorithmic}
\caption{Algorithm for finding a matching with the minimum number of blocking pairs}
\label{alg:minbp}
\end{algorithm}

\begin{theorem}
\label{thm:xp}
    Given an {\sc sf} instance with $n$ agents, Algorithm \ref{alg:minbp} returns a \emph{MADI}$_{\{c\}}$ matching in $O(n^{b+2})$ time, where $b=\min_{M\in\mathcal M}\vert be(M)\vert$.
\end{theorem}
\begin{proof}
    We showed in the proof of Theorem \ref{thm:madichard} that a matching is MADI$_{\{c\}}$ if and only if, for all matchings $M'$ of $I$, it is true that $\vert bp(M)\vert\leq \vert bp(M')\vert$. Thus, it suffices to prove that the matching $M$ returned by Algorithm \ref{alg:minbp} satisfies $\vert bp(M)\vert\leq \vert bp(M')\vert$. Clearly, the algorithm terminates, because when $k=\vert A\vert(\vert A\vert-1)$ we make all pairs of agents unacceptable, so a stable matching exists trivially. Now, let $k^*$ be the value for $k$ and let $B^*$ be the set $B$ when the algorithm terminates. Notice that, clearly, $bp(M)\subseteq B^*$, i.e., $\vert bp(M)\vert \leq \vert B^*\vert=k^*$. Now suppose that there exists some matching $M'$ of $I$ such that $\vert bp(M')\vert<k^*$. Then $M'$ is stable in the instance $I'$ in which all pairs of agents in $bp(M')$ are unacceptable. Hence, the algorithm would have terminated at $k'=\vert bp(M')\vert<k^*$ instead, a contradiction. Thus, $k^*$ is minimal, i.e., $k^*=\min_{M\in\mathcal M}\vert bp(M)\vert$ and therefore $M$ is MADI$_{\{c\}}$ as required.

    For the time complexity, let $k^*$ be the value of $k$ when the algorithm terminates. It is easy to see that for $k^*=0$, the main loop is executed exactly once as $B=\varnothing$. For all other values of $k^*$, the total number of iterations of the main loop, taken over the entire algorithm's execution, is at most $\sum_{1\leq r\leq k^*}\binom{n(n-1)}{r}$ many times. We can apply the loose upper bound $\binom{n(n-1)}{r}\leq (n(n-1))^r$ and conclude that $\sum_{1\leq r\leq k^*} (n(n-1))^r=\frac{(n(n-1))^{k^*+1}-1}{n(n-1)-1}$ by the geometric series. Furthermore, the bound $\frac{(n(n-1))^{k^*+1}-1}{n(n-1)-1}\leq \frac{(n(n-1))^{k^*+1}}{n(n-1)-1}\leq 2(n(n-1))^{k^*}=O((n(n-1))^{k^*})=O(n^{2k^*})$ applies. Now, in each execution of the inner for loop, making a copy $I'$ of $I$ and making some agents unacceptable requires $O(n^2)$ time overall, and subsequently executing the algorithm by \citet{Irving2007} also requires $O(n^2)$ time. Thus, the algorithm terminates after at most $O(n^{2k^*+2})$ steps. Given that $k^*=\min_{M\in\mathcal M}\vert bp(M)\vert=\frac{1}{2}\min_{M\in\mathcal M}\vert be(M)\vert=\frac{1}{2}b$, the stated time complexity follows.
\end{proof}

\subsection{Integer Linear Programs}
\label{sec:ilp}

Our final technical contribution consists of two compact integer linear programming (ILP) formulations that can be used to find MIDI$_{\{c\}}$ and MADI$_{\{c\}}$ matchings. We start by providing an ILP for the computation of a MIDI$_{\{c\}}$ matching, leveraging the property from Proposition \ref{prop:madicclaim2} that states that it suffices to minimise the maximum number of blocking pairs that any agent is contained in among all agents. Let $I=(A,\succ,c)$ be a problem instance and consider the ILP below, where $x_{ij},b_{ij}$ and $w_{ij}$ are binary indicator variables that indicate for the pair of agents $a_i-a_j$ whether they are matched (in which case $x_{ij}=1$), they are blocking (in which case $b_{ij}=1$), and whether $a_i$ is matched up to capacity to agents at least least as good as $a_j$ (in which case $w_{ij}=1$).

\begin{align} 
\min &\; r \\
\text{s.t.} \sum_{a_j\in A\setminus\{a_i\}}x_{ij} &\leq c_i & \forall a_i\in A \\ 
x_{ij}&=x_{ji}   & \forall a_i,a_j\in A \\
b_{ij}&=b_{ji}   & \forall a_i,a_j\in A \\
\sum_{a_k:a_k\succeq_ia_j}x_{ik} &\geq c_iw_{ij} &\forall a_i,a_j\in A\\
w_{ij}+w_{ji} + b_{ij} &\geq 1 &\forall a_i,a_j\in A\\
\sum_{a_j\in A\setminus\{a_i\}}b_{ij} &\leq r & \forall a_i\in A\\
x_{ij}, b_{ij}, w_{ij} &\in\{0,1\} &\forall a_i,a_j\in A\\
r &\in\mathbb Z^{\geq 0}
\end{align}

\begin{theorem}
    A matching is \emph{MIDI}$_{\{c\}}$ if and only if it corresponds to an optimal solution to the ILP given by $(1)$-$(9)$.
\end{theorem}
\begin{proof}
    Let $M$ be a MIDI$_{\{c\}}$ matching. For any $\{a_i,a_j\}\in M$, set $x_{ij}=x_{ji}=1$ and, for any $\{a_i,a_j\}\notin M$, set $x_{ij}=x_{ji}=0$. Furthermore, for any $\{a_i,a_j\}\in bp(M)$, set $b_{ij}=b_{ji}=1$ and, for any $\{a_i,a_j\}\notin bp(M)$, set $b_{ij}=b_{ji}=0$. Also, let $w_{ij}=1$ if and only if $a_i$ has $c_i$ many matches who they rank at least as highly as $a_j$. Finally, let $r=\max_{a_i\in A}\vert bp_i(M) \vert$. Then clearly our assignment represents an optimal solution to the above ILP. 

    Now, conversely, let $\langle r,\mathbf{x},\mathbf{b},\mathbf{w}\rangle$ be an optimal solution to the ILP. Let $M$ consist of all pairs of agents $\{a_i,a_j\}$ such that $x_{ij}=1$. Then it follows easily from the constraints that $M$ respects all agent capacities and minimises the maximum number of blocking pairs per agent. We showed in the proof of Theorem \ref{thm:midichard} that a matching is MIDI$_{\{c\}}$ if and only if, for all matchings $M'$ of $I$, it is true that $\max_{a_i\in A}\vert bp_i(M)\vert\leq \max_{a_i\in A}\vert bp_i(M')\vert$. Thus, the correctness of the ILP follows.
\end{proof}

Now let us turn to the computation of a MADI$_{\{c\}}$ matching. We can keep constraints $(2)$-$(6)$ as well as constraint $(8)$ from the previous ILP, discard the variable $r$, and use the following objective function instead:
\begin{align} 
\min &\; \sum_{a_i\in A}\sum_{a_j\in A\setminus\{a_i\}}b_{ij} 
\end{align}

\begin{theorem}
    A matching is \emph{MADI}$_{\{c\}}$ if and only if it corresponds to an optimal solution to the ILP given by $(2)$-$(6)$, $(8)$ and $(10)$.
\end{theorem}
\begin{proof}
    Let $M$ be a MADI$_{\{c\}}$ matching. For any $\{a_i,a_j\}\in M$, set $x_{ij}=x_{ji}=1$ and, for any $\{a_i,a_j\}\notin M$, set $x_{ij}=x_{ji}=0$. Furthermore, for any $\{a_i,a_j\}\in bp(M)$, set $b_{ij}=b_{ji}=1$ and, for any $\{a_i,a_j\}\notin bp(M)$, set $b_{ij}=b_{ji}=0$. Also, let $w_{ij}=1$ if and only if $a_i$ has $c_i$ many matches who they rank at least as highly as $a_j$. Then clearly our assignment represents an optimal solution to the above ILP. 

    Now, conversely, let $\langle\mathbf{x},\mathbf{b},\mathbf{w}\rangle$ be an optimal solution to the ILP. Let $M$ consist of all pairs of agents $\{a_i,a_j\}$ such that $x_{ij}=1$. Then it follows easily from the constraints that $M$ respects all agent capacities and minimises the number of blocking pairs. We showed in the proof of Theorem \ref{thm:madichard} that a matching is MADI$_{\{c\}}$ if and only if, for all matchings $M'$ of $I$, it is true that $\vert bp(M)\vert\leq \vert bp(M')\vert$. Thus, the correctness of the ILP follows.
\end{proof}

\section{Experiments on Optimal Modifications and Blocking Entries}
\label{sec:exp}

In this paper, we set out to investigate the incentive and optimality landscape surrounding near-feasible stable matchings. We laid out a new framework to analyse these issues, and analysed the computational complexity of optimisation problems within this framework. Now, we will aim to answer the following questions experimentally: what is the average minimum number of capacity modifications we need to make to random {\sc sf} instances to arrive at a (modified) solvable instance? Similarly, what is the average minimum number of blocking entries we need to accept for random {\sc sf} instances?

We provide answers to these questions using synthetically generated instances. As is common practice in synthetic experiments for matching under preferences (e.g., see references \cite{delorme2019mathematical, pettersson2021improving,mertens15random,glitzner2025empirics}), we consider preferences over agents sampled uniformly at random. We characterise instances by their number of agents $n$ and a fixed capacity $c$ that applies to all agents, and average our results over 1000 instances per $(n,c)$ pair and algorithm. We chose $n\in\{10,12,14,\dots,40\}$ and $c\in\{1,3,5,7\}$. All implementations were written in Python, and all computations were performed on the {\tt fataepyc} cluster. We used four compute nodes at a time, each equipped with dual AMD EPYC 7643 CPUs and 2TB RAM. All code is publicly available; see reference \cite{experimentsCode}. The generated instances are seeded and can thus be easily replicated. All ILP models are solved using the PuLP package with the GurobiPy solver. We note that the execution of 83 out of the 192,000 total configuration and algorithm combinations did not terminate within a two-hour time limit, even after rerunning the experiment, and we excluded these executions from our analysis.

First, we report on the average number of optimal capacity modifications (in any direction) and on the average minimum amount of instability that needs to be accepted when capacity increases are permitted. We showed in Theorem \ref{thm:mam} that the minimum total number of capacity modifications is completely determined by the structural property $\mathcal{O}_I^{\geq 3}$ of a given {\sc sf} instance $I$. In Theorem \ref{thm:madigeneral}, we showed that the same is true for the minimum number of blocking entries when capacity increases are permitted. Thus, we measured $\vert \mathcal{O}_I^{\geq 3}\vert$ across our range of different sizes and capacity functions. Our results are summarised in Figure \ref{fig:sizevextensions}. The first observation to make is that the average number of capacity extensions appears to grow linearly with the number of agents. However, we also note that the absolute number of capacity changes is very small on average -- in fact, the average is below 1, which is due to the presence of solvable instances that require no capacity modifications. 

\begin{figure}[!htb]
    \centering
    \begin{tikzpicture}
        \begin{axis}[
            width=.99\textwidth,
            height=5cm,
            ylabel={Extensions},
            ymin=0.1,
            ymax=0.55,
            xmin=9,
            xlabel={Number of agents},
            xmax=41,
            grid=both,
            grid style={dashed, gray!30},
            cycle list name=color list,
            every axis plot/.append style={thick},
            title={},
            title style={
                yshift=-1.5ex
            },
            legend cell align={left},
            legend style={
                at={(1,0)},
                anchor=south east,
                column sep=1ex,
            },
        ]
        \addplot[acmDarkBlue, mark=*] table [x=n, y=c1] {data/capext.txt};
        \addplot[acmGreen, mark=square*] table [x=n, y=c2] {data/capext.txt};
        \addplot[acmPink, mark=triangle*] table [x=n, y=c3] {data/capext.txt};
        \addplot[acmOrange, mark=diamond*] table [x=n, y=c4] {data/capext.txt};
        \addlegendentry{$c=1$}
        \addlegendentry{$c=3$}
        \addlegendentry{$c=5$}
        \addlegendentry{$c=7$}
        \end{axis}
    \end{tikzpicture}
    \caption{Average number of capacity extensions}
    \label{fig:sizevextensions}
\end{figure}
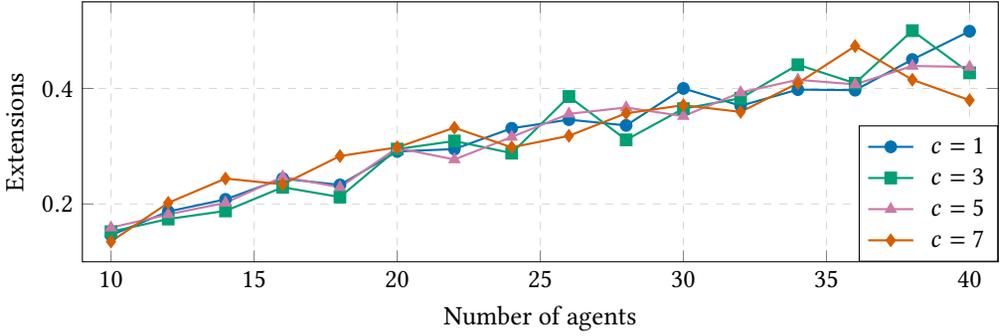

Figure \ref{fig:sizevextensions2} shows the same results but adjusted for solvable instances, i.e., removing all 0 entries from the averages. Still, the absolute number of capacity adjustments is small on average for all four capacity functions. The maximum number of capacity extensions we found among all generated instances is 4. Simultaneously, this implies that, on average, very few blocking entries need to be accepted.

\begin{figure}[!htb]
    \centering
    \begin{tikzpicture}
        \begin{axis}[
            width=.99\textwidth,
            height=5cm,
            ylabel={Extensions},
            ymin=1,
            ymax=2.1,
            xmin=9,
            xlabel={Number of agents},
            xmax=41,
            grid=both,
            grid style={dashed, gray!30},
            cycle list name=color list,
            every axis plot/.append style={thick},
            title={},
            title style={
                yshift=-1.5ex
            },
            legend cell align={left},
            legend style={
                at={(1,0)},
                anchor=south east,
                column sep=1ex,
            },
        ]
        \addplot[acmDarkBlue, mark=*] table [x=n, y=c1] {data/capext2.txt};
        \addplot[acmGreen, mark=square*] table [x=n, y=c2] {data/capext2.txt};
        \addplot[acmPink, mark=triangle*] table [x=n, y=c3] {data/capext2.txt};
        \addplot[acmOrange, mark=diamond*] table [x=n, y=c4] {data/capext2.txt};
        \addlegendentry{$c=1$}
        \addlegendentry{$c=3$}
        \addlegendentry{$c=5$}
        \addlegendentry{$c=7$}
        \end{axis}
    \end{tikzpicture}
    \caption{Average number of capacity extensions for unsolvable instances}
    \label{fig:sizevextensions2}
\end{figure}
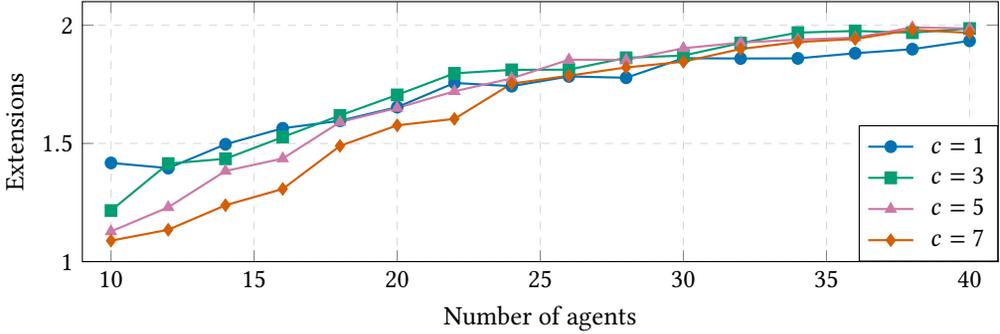

Next, we report on minimum instability when capacity increases are not permitted; these are precisely the {\sf NP-hard} cases we identified in Section \ref{sec:midimadic} and designed exponential-time algorithms for in Section \ref{sec:exact}. For our experiments, we used the ILP models outlined in Section \ref{sec:ilp}. We measured the minimum number of blocking entries in total and the minimum maximum number of blocking entries per agent across a range of different sizes and capacity functions. Surprisingly, all generated instances admitted a matching with either 0 or 1 blocking pairs for every agent, i.e., here we could guarantee a MIDI$_{\{c\}}$ solution requiring at most one blocking entry per agent. However, the aggregate guarantees of these matchings are much weaker -- Figure \ref{fig:bpind} shows that these MIDI$_{\{c\}}$ matchings admit up to 4 blocking pairs on average (when $n=40$ and $c=7$), and likely more when either $n$ or $c$ is larger.

\begin{figure}[!htb]
    \centering
    \begin{tikzpicture}
        \begin{axis}[
            width=.99\textwidth,
            height=5cm,
            ylabel={Blocking pairs},
            ymin=1,
            ymax=4,
            xmin=9,
            xlabel={Number of agents},
            xmax=41,
            grid=both,
            grid style={dashed, gray!30},
            cycle list name=color list,
            every axis plot/.append style={thick},
            title={},
            title style={
                yshift=-1.5ex
            },
            legend cell align={left},
            legend style={
                at={(1,0)},
                anchor=south east,
                column sep=1ex,
            },
        ]
        \addplot[acmDarkBlue, mark=*] table [x=n, y=c1] {data/midibp.txt};
        \addplot[acmGreen, mark=square*] table [x=n, y=c2] {data/midibp.txt};
        \addplot[acmPink, mark=triangle*] table [x=n, y=c3] {data/midibp.txt};
        \addplot[acmOrange, mark=diamond*] table [x=n, y=c4] {data/midibp.txt};
        \addlegendentry{$c=1$}
        \addlegendentry{$c=3$}
        \addlegendentry{$c=5$}
        \addlegendentry{$c=7$}
        \end{axis}
    \end{tikzpicture}
    \caption{Average aggregate number of blocking pairs of MIDI$_{\{c\}}$ matchings for unsolvable instances}
    \label{fig:bpind}
\end{figure}
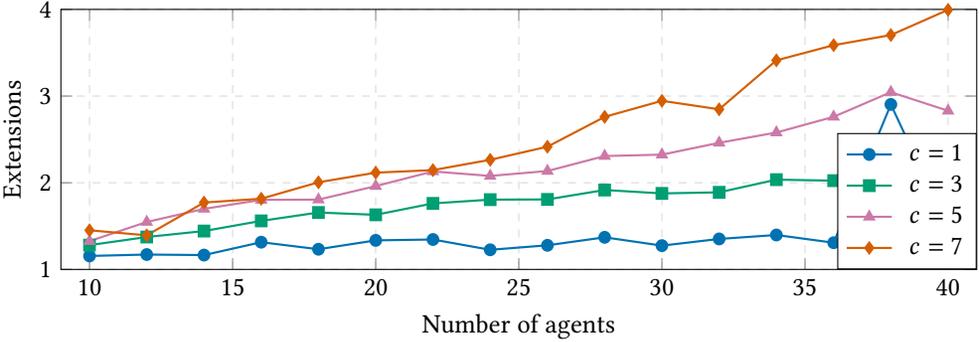

On the other hand, our results for MADI$_{\{c\}}$ matchings show that the average aggregate number of blocking pairs required is no more than 1.003 for any capacity function (averaged across all values for $n$ but restricted to unsolvable instances, i.e., requiring at least one blocking pair), so it must be possible in almost all cases to match the agents in a way such that only two agents are in one blocking pair (or vice versa one blocking entry), and all others are in none. We note that none of the instances generated required more than two blocking pairs in total.

Overall, these experiments suggest that random {\sc sf} instances require either very few capacity modifications or very few blocking entries, which is positive news for practical applications as well as the applicability of the {\sf XP} Algorithm \ref{alg:minbp}, which runs efficiently whenever the optimal value is small.

\section{Conclusion and Future Directions}
\label{sec:conclusion}

We studied the deviation incentives and complexity of near-feasible stable matchings and their computation in the non-bipartite, capacitated stable matching setting. Our main takeaways can be summarised as follows:
\begin{itemize}[leftmargin=1.5em]
    \item a well-balanced alternative capacity function leading to a solvable {\sc sf} instance that requires at most one capacity modification per agent and simultaneously minimises the total capacity modifications always exists and can be computed in polynomial time;
    \item a well-balanced matching that requires at most one blocking entry per agent and simultaneously minimises the total deviation incentive always exists and can be computed in polynomial time. Furthermore, this solution even satisfies the principles of minimal individual and minimal aggregate capacity modification;
    \item from the perspective of deviation incentives, increasing agent capacities is always weakly better than decreasing;
    \item when original agent capacities must not be violated, then minimising blocking entries is equivalent to minimising blocking pairs, which is known to be {\sf NP-hard} even in the uncapacitated agent setting;
    \item random experiments suggest that both the minimum number of capacity modifications required to arrive at a solvable instance and the minimum instability that needs to be accepted when capacities cannot be modified are small in practice.
\end{itemize}
Table \ref{table:results} gives a more concise overview of the computational complexity results we established.

\begin{table*}[!tbh]
    \centering
    \caption{Complexity of computing MIDI or MADI matchings with respect to various types of capacity restrictions or modifications. The {\sf XP} results is with respect to the optimum parameter value.}
    \small
    \begin{tabular}{c c c c c c c c c}
        \toprule
        & \multicolumn{2}{c}{} & \multicolumn{2}{c}{\textbf{MIM}} & \multicolumn{2}{c}{\textbf{MAM}} \\
        \cmidrule(lr){4-5} \cmidrule(lr){6-7}
        & \textbf{$\{c\}$} & \textbf{$\infty$} 
        & \textbf{MIM$_\text{+}$} & \textbf{MIM$_{\pm}$}
        & \textbf{MAM$_\text{+}$} & \textbf{MAM$_{\pm}$} \\
        \midrule
        \textbf{MIDI} & {\sf NP-hard} [T \ref{thm:midichard}]  & {\sf P} [P \ref{prop:nearfeasiblemidi}] & {\sf P} [P \ref{prop:nearfeasiblemidi}] & {\sf P} [P \ref{prop:nearfeasiblemidi}] & {\sf P} [P \ref{prop:nearfeasiblemidi}] & {\sf P} [P \ref{prop:nearfeasiblemidi}] \\
        \textbf{MADI} & {\sf NP-hard} [T \ref{thm:madichard}], {\sf XP} [T \ref{thm:xp}] & {\sf P} [P \ref{prop:nearfeasiblemadi}] & {\sf P} [P \ref{prop:nearfeasiblemadi}] & {\sf P} [P \ref{prop:nearfeasiblemadi}] & {\sf P} [P \ref{prop:nearfeasiblemadi}] & {\sf P} [P \ref{prop:nearfeasiblemadi}] \\
        \bottomrule
    \end{tabular}
    \label{table:results}
\end{table*}

\subsection{Limitations and More Complex Agent Incentive Models}
\label{sec:alternative}
 
We motivated the notions MIDI and MADI as objectives that minimise deviation incentives of agents with respect to matchings. 
However, one of our central assumptions was that agent incentives are uniform in essentially every way: our main tool, the notion of a blocking entry, measures an agents' incentive to deviate the same regardless of whether they had a small or large original capacity, whether they are assigned more agents than desired or they would rather be matched to more agents, and we also do not consider how high in the ranking the blocking entries are. 

To illustrate one of these issues, consider the following scenario: let $I$ be an {\sc sf} instance such that any GSP of $I$ contains a cycle $C$ of odd length at least 3. Let $a_i,a_j$ be two distinct agents that are both present in $C$ and suppose that $c_i\gg c_j$. Our previously studied versions of the \texttt{NearFeasible$_\pm$} algorithm that we showed to perform optimally with respect to computing a MIM or a MAM could pick either of $a_i,a_j$ (or any other agent present in $C$) to change their capacity. In fact, with regard to the blocking entries admitted by the computed solution, the choice between $a_i$ and $a_j$ might not make a difference. However, from an agent incentive perspective, it would seem much more natural to change the capacity $c_i$ rather than $c_j$ because the \emph{marginal} change is much lower.

To illustrate another of the mentioned issues, which is inherent to the definitions of MIDI and MADI, consider the \emph{uniform treatment} of blocking entries. In practice, an agent's incentive to deviate if they could be matched to their first choice instead of their last choice is arguably much higher than the incentive to deviate if they could be matched to someone far down their preference list instead. This is especially true when considering that a deviation can cause subsequent unravelling, so an agent might risk not just disadvantaging others by deviating, but may also risk losing other good matches that they are assigned to.

These issues motivate more complex agent incentive models. For example, one could imagine that each agent $a_i\in A$ is equipped with a cost function that indicates the cost of any given preference list entry being a blocking entry (either of type (i) or of type (ii)), and the objective could become to minimise the global or individual sum of blocking entry costs. A simpler version that does not specifically depend on which entry is blocking would be to simply weight blocking entries that stem from blocking pairs differently to blocking entries that stem from capacity violations. However, even the most natural problem of minimising this weighted cost either globally or per agent leads to strong computational intractability. In particular, let the blocking entry cost of every agent be $\alpha\in\mathbb R_+$ times the number of blocking entries that stem from a blocking pair plus $\beta\in\mathbb R_+$ times the number of blocking entries that stem from a capacity violation. Then for any $\alpha$ and $\beta$ such that $\alpha<\beta$, minimising the maximum individual blocking entry cost is {\sf NP-hard}, because it would let us decide whether there exists a matching where no agent is in more than one blocking pair (this is the case when the optimum value is at most $\alpha$) or whether every matching requires at least one agent to be in more than one blocking pair (this is the case when the optimum value is strictly more than $\alpha$). This follows from a similar simple reduction as in the proof of Theorem \ref{thm:midichard}. Similarly, minimising the aggregate blocking entry cost is {\sf NP-hard} in general, as for any given $\alpha$ we can simply pick $\beta\geq n^2\alpha$, in which case the matching that minimises the blocking entry cost is a matching without capacity violations that minimises the number of blocking pairs (because every matching admits at most $n(n-1)$ blocking pairs), which is an {\sf NP-hard} problem (see the proof of Theorem \ref{thm:madichard}).

\subsection{Future Research Directions}
Several intriguing directions for further research remain. We introduced new optimality measures for matchings in the many-to-many setting; it would be interesting to evaluate how other existing approaches in many-to-one and many-to-many models perform under these criteria, and whether certain techniques systematically minimise agents' incentives to deviate. From an economic perspective, analysing real-world matching markets through the lens of blocking entries and deviation incentives could yield further insights. On the algorithmic side, identifying special cases where MIDI$_{\{c\}}$ and MADI$_{\{c\}}$ matchings can be computed efficiently remains an open question, though strong intractability results from almost-stable matchings suggest fundamental limitations. On the experimental side, considering more diverse statistical distributions in the preference generation (such as those outlined by Boehmer et al. \cite{boehmer2024map}), or using real-world data, would be interesting. Finally, identifying more complex agent incentive models where general computational tractability can be achieved is a very interesting avenue for future research.

\paragraph{Funding}
Frederik Glitzner is supported by a Minerva Scholarship from the School of Computing Science, University of Glasgow. 


\paragraph{Acknowledgements}
The author would like to thank David Manlove for very helpful comments and discussions, and the anonymous AAMAS and MATCH-UP reviewers for helpful comments and suggestions.


\bibliographystyle{reference_format}
\bibliography{papers}

\clearpage
\appendix 

\section{Algorithms}
\label{sec:appendixAlgos}

In this section, we provide the algorithms {\sf NearFeasible}$_\pm$, {\sf NearFeasible}$_+$ and {\sf NearFeasible}$_-$ in full, but they can be derived from \cite{glitzner2025unsolvabilitymanytomanynonbipartitestable} as described in the main body of the paper.

\subsection{Algorithm {\sf NearFeasible}$_\pm$}

Pseudocode for {\sf NearFeasible}$_\pm$ is provided in Algorithm \ref{alg:constructnearfeasibleplusminus}. The procedure works as follows: Given an {\sc sf} instance $I=(A,\succ,c)$ and a GSP $\Pi$ of $I$, we decompose all cycles of length longer than 2 into transpositions, turn all transpositions into matches, and alternatively increasing and decreasing the capacity of one agent in each cycle of odd length at least 3 (we denote the set of such cycles by $\mathcal{O}_I^{\geq 3}$ and denote the set of modified agents by $O$). This ensures that we increase the capacity of $\left\lceil\frac{\vert O\vert}{2}\right\rceil$ agents by 1 and decrease the capacity of the remaining agents of $O$ by 1 to arrive at $I'=(A,\succ,c')$.

\begin{algorithm}[hbt]
\renewcommand{\algorithmicrequire}{\textbf{Input:}}
\renewcommand{\algorithmicensure}{\textbf{Output:}}

    \begin{algorithmic}[1]

    \Require{$I=(A,\succ,c)$ : an unsolvable {\sc sf} instance; $\Pi$ : a GSP1 of $I$}
    \Ensure{$c'$ : a new capacity function; $M$ : a stable matching in $I'=(A,\succ,c)$; $O$ : the set of agents with modified capacities in $c'$}

    \State $M \gets \varnothing$
    \State $i \gets 0$
    \State $c'\gets c$
    \State $O\gets \varnothing$

    \For{cycle $\Pi_k\in \Pi$}
        \If{$\vert A_k\vert=2$} \Comment{Add the transpositions as matches}
            \State $M$.add$(\{A_k[0], A_k[1]\})$
        \EndIf

        \State $m\gets\vert A_k\vert$
        
        \If{$m>2$} 
            \If{$m$ is even} \Comment{Decompose longer even-length cycles}
                \State $M$.add$(\{A_k[0], A_k[1]\},\{A_k[2], A_k[3]\}\dots\{A_k[m-2], A_k[m-1]\})$
            \Else 
                \If{$i$ is even} \Comment{Decompose longer odd-length cycles by increasing capacity}
                    \State $M$.add$(\{A_k[0], A_k[1]\},\{A_k[2], A_k[3]\}\dots\{A_k[m-1], A_k[0]\})$
                    \State $c'_{A_k[0]}=c_{A_k[0]}+1$
                    \State $O$.add$(A_k[0])$
                \Else \Comment{Decompose longer odd-length cycles by decreasing capacity}
                    \State $M$.add$(\{A_k[1], A_k[2]\},\{A_k[3], A_k[4]\}\dots\{A_k[m-2], A_k[m-1]\})$
                    \State $c'_{A_k[0]}=c_{A_k[0]}-1$
                    \State $O$.add$(A_k[0])$
                \EndIf
                \State $i \gets i+1$
            \EndIf
        \EndIf
    \EndFor
    \State\Return{$c', M,O$}

    \end{algorithmic}
    \caption{Constructs a near-feasible solvable instance $I'$ and a corresponding stable matching $M$ from a GSP $\Pi$ of $I$}
    \label{alg:constructnearfeasibleplusminus}
\end{algorithm}

Glitzner and Manlove \cite{glitzner25sagt} proved the following correctness result (in Theorem 15 of their paper).

\begin{theorem}
    Let $I=(A,\succ,c)$ be an unsolvable {\sc sf} instance with $n\geq 3$ agents and let $\Pi$ be a GSP of $I$. Then Algorithm \ref{alg:constructnearfeasibleplusminus} finds a modified instance $I'=(A,\succ,c')$ and a stable matching $M$ of $I'$ in $O(n^2)$ time ($O(n^3$ if $\Pi$ is initially unknown) such that for all $a_i\in A$, we have that $c_i'\in\{c_i-1, c_i,c_i+1\}$, and $\sum_{a_i\in A}(c_i'-c_i)=\vert\mathcal{O}_I^{\geq3}\vert\mod 2\,\leq 1$. Furthermore, $\sum_{a_i\in A}\vert c_i'-c_i\vert =  \vert\mathcal{O}_I^{\geq3}\vert\leq \frac{n}{3}$, where $\mathcal{O}_I^{\geq3}$ are the cycles of $\Pi$ with odd length at least 3.
\end{theorem}

\subsection{Algorithm {\sf NearFeasible}$_+$}

Pseudocode for {\sf NearFeasible}$_+$ is provided in Algorithm \ref{alg:constructnearfeasibleplus}. The procedure differs from {\sf NearFeasible}$_\pm$ as follows: We no longer keep track of $i$ and, instead, 
always decompose longer odd-length cycles by increasing capacity. In other words, we always apply the ``if'' case in line 14 of Algorithm \ref{alg:constructnearfeasibleplusminus}.

\begin{algorithm}[hbt]
\renewcommand{\algorithmicrequire}{\textbf{Input:}}
\renewcommand{\algorithmicensure}{\textbf{Output:}}

    \begin{algorithmic}[1]

    \Require{$I=(A,\succ,c)$ : an unsolvable {\sc sf} instance; $\Pi$ : a GSP1 of $I$}
    \Ensure{$c'$ : a new capacity function; $M$ : a stable matching in $I'=(A,\succ,c)$; $O$ : the set of agents with modified capacities in $c'$}

    \State $M \gets \varnothing$
    \State $c'\gets c$
    \State $O\gets \varnothing$

    \For{cycle $\Pi_k\in \Pi$}
        \If{$\vert A_k\vert=2$} \Comment{Add the transpositions as matches}
            \State $M$.add$(\{A_k[0], A_k[1]\})$
        \EndIf

        \State $m\gets\vert A_k\vert$
        
        \If{$m>2$} 
            \If{$m$ is even} \Comment{Decompose longer even-length cycles}
                \State $M$.add$(\{A_k[0], A_k[1]\},\{A_k[2], A_k[3]\}\dots\{A_k[m-2], A_k[m-1]\})$
            \Else \Comment{Decompose longer odd-length cycles by increasing capacity}
                \State $M$.add$(\{A_k[0], A_k[1]\},\{A_k[2], A_k[3]\}\dots\{A_k[m-1], A_k[0]\})$
                \State $c'_{A_k[0]}=c_{A_k[0]}+1$
                \State $O$.add$(A_k[0])$
            \EndIf
        \EndIf
    \EndFor
    \State\Return{$c', M,O$}

    \end{algorithmic}
    \caption{Constructs a near-feasible solvable instance $I'$ and a corresponding stable matching $M$ from a GSP $\Pi$ of $I$}
    \label{alg:constructnearfeasibleplus}
\end{algorithm}

The following result follows easily from the proof of Theorem 15 in \cite{glitzner25sagt} using Lemma 3 of \cite{glitzner25sagt}.

\begin{theorem}
    Let $I=(A,\succ,c)$ be an unsolvable {\sc sf} instance with $n\geq 3$ agents and let $\Pi$ be a GSP of $I$. Then Algorithm \ref{alg:constructnearfeasibleplus} finds a modified instance $I'=(A,\succ,c')$ and a stable matching $M$ of $I'$ in $O(n^2)$ time ($O(n^3)$ if $\Pi$ is initially unknown) such that for all $a_i\in A$, we have that $c_i'\in\{c_i,c_i+1\}$, and $\sum_{a_i\in A}(c_i'-c_i)=\vert\mathcal{O}_I^{\geq3}\vert$. Furthermore, $\sum_{a_i\in A}\vert c_i'-c_i\vert =  \vert\mathcal{O}_I^{\geq3}\vert\leq \frac{n}{3}$, where $\mathcal{O}_I^{\geq3}$ are the cycles of $\Pi$ with odd length at least 3.
\end{theorem}

\subsection{Algorithm {\sf NearFeasible}$_-$}

Pseudocode for {\sf NearFeasible}$_-$ is provided in Algorithm \ref{alg:constructnearfeasibleminus}. The procedure differs from {\sf NearFeasible}$_+$ as follows: Instead of always decomposing longer odd-length cycles by increasing capacity, we always decompose longer odd-length cycles by decreasing capacity. In other words, we always apply the ``if'' case in line 18 of Algorithm \ref{alg:constructnearfeasibleplusminus}.

\begin{algorithm}[hbt]
\renewcommand{\algorithmicrequire}{\textbf{Input:}}
\renewcommand{\algorithmicensure}{\textbf{Output:}}

    \begin{algorithmic}[1]

    \Require{$I=(A,\succ,c)$ : an unsolvable {\sc sf} instance; $\Pi$ : a GSP1 of $I$}
    \Ensure{$c'$ : a new capacity function; $M$ : a stable matching in $I'=(A,\succ,c)$; $O$ : the set of agents with modified capacities in $c'$}

    \State $M \gets \varnothing$
    \State $c'\gets c$
    \State $O\gets \varnothing$

    \For{cycle $\Pi_k\in \Pi$}
        \If{$\vert A_k\vert=2$} \Comment{Add the transpositions as matches}
            \State $M$.add$(\{A_k[0], A_k[1]\})$
        \EndIf

        \State $m\gets\vert A_k\vert$
        
        \If{$m>2$} 
            \If{$m$ is even} \Comment{Decompose longer even-length cycles}
                \State $M$.add$(\{A_k[0], A_k[1]\},\{A_k[2], A_k[3]\}\dots\{A_k[m-2], A_k[m-1]\})$
            \Else \Comment{Decompose longer odd-length cycles by decreasing capacity}
                \State $M$.add$(\{A_k[1], A_k[2]\},\{A_k[3], A_k[4]\}\dots\{A_k[m-2], A_k[m-1]\})$
                \State $c'_{A_k[0]}=c_{A_k[0]}-1$
                \State $O$.add$(A_k[0])$
            \EndIf
        \EndIf
    \EndFor
    \State\Return{$c', M,O$}

    \end{algorithmic}
    \caption{Constructs a near-feasible solvable instance $I'$ and a corresponding stable matching $M$ from a GSP $\Pi$ of $I$}
    \label{alg:constructnearfeasibleminus}
\end{algorithm}

The following result again follows easily from the proof of Theorem 15 in \cite{glitzner25sagt} using Lemma 3 of \cite{glitzner25sagt}.

\begin{theorem}
    Let $I=(A,\succ,c)$ be an unsolvable {\sc sf} instance with $n\geq 3$ agents and let $\Pi$ be a GSP of $I$. Then Algorithm \ref{alg:constructnearfeasibleminus} finds a modified instance $I'=(A,\succ,c')$ and a stable matching $M$ of $I'$ in $O(n^2)$ time ($O(n^3)$ if $\Pi$ is initially unknown) such that for all $a_i\in A$, we have that $c_i'\in\{c_i,c_i-1\}$, and $\sum_{a_i\in A}(c_i'-c_i)=-\vert\mathcal{O}_I^{\geq3}\vert$. Furthermore, $\sum_{a_i\in A}\vert c_i'-c_i\vert =  \vert\mathcal{O}_I^{\geq3}\vert\leq \frac{n}{3}$, where $\mathcal{O}_I^{\geq3}$ are the cycles of $\Pi$ with odd length at least 3.
\end{theorem}

\end{document}